\documentclass[]{article}
\usepackage[margin=1in]{geometry}
\usepackage{amsmath}
\usepackage{ amssymb }
\usepackage{ amsthm }
\usepackage{ bbold }
\usepackage{newenviron}
\usepackage{booktabs} 
\usepackage{float}
\usepackage{threeparttable}
\usepackage{fouriernc} 
\newcommand{\Q}{\widehat{Q}}

\newcommand{\z}{{\mathbf{z}}}
\newcommand{\zbar}{\overline{\mathbf{z}}}
\usepackage{makecell}
\usepackage{multirow}
\newcommand{\aA}{\overline{a}_A}
\newcommand{\bB}{\overline{b}_B}
\newcommand{\haA}{\overline{ha}_A}
\newcommand{\hbB}{\overline{hb}_B}

\newcommand{\asA}{\overline{a^*}_A}
\newcommand{\bsB}{\overline{b^*}_B}
\newcommand{\zA}{\overline{\mathbf{z}}_A}
\newcommand{\zB}{\overline{\mathbf{z}}_B}
\newcommand{\azA}{\overline{a\mathbf{z}}_A}
\newcommand{\bzB}{\overline{b\mathbf{z}}_B}
\newcommand{\az}{\overline{a\mathbf{z}}}
\newcommand{\bz}{\overline{b\mathbf{z}}}
\newcommand{\aszA}{\overline{a^*\mathbf{z}}_A}
\newcommand{\bszB}{\overline{b^*\mathbf{z}}_B}
\newcommand{\asz}{\overline{a^*\mathbf{z}}}
\newcommand{\bsz}{\overline{b^*\mathbf{z}}}
\newcommand{\zz}{\overline{\mathbf{z}\mathbf{z}'}}
\newcommand{\zzA}{\overline{\mathbf{z}\mathbf{z}'}_A}
\newcommand{\zzB}{\overline{\mathbf{z}\mathbf{z}'}_B}
\newcommand{\pA}{{p}_A}
\newcommand{\pB}{{p}_B}
\newcommand{\E}{\mathbf{E}}
\newcommand{\N}{\widehat{N}}
\newcommand{\D}{\widehat{D}}
\newcommand{\zBar}{\overline{\mathbf{z}}}

\newcommand{\ateAdj}{\widehat{ATE}_{adj}}

\usepackage{hyperref}
\usepackage{lscape}
\usepackage{thmtools}
\newcommand\fixstatement[2][\proofname\space of]{%
	\ifcsname thmt@original@#2\endcsname
	\AtEndEnvironment{#2}{%
		\xdef\pat@label{\expandafter\expandafter\expandafter
			\@fourthoffour\csname thmt@original@#2\endcsname\space\@currentlabel}%
		\xdef\pat@proofof{\@nameuse{pat@proofof@#2}}%
	}%
	\else
	\AtEndEnvironment{#2}{%
		\xdef\pat@label{\expandafter\expandafter\expandafter
			\@fourthoffour\csname #1\endcsname\space\@currentlabel}%
		\xdef\pat@proofof{\@nameuse{pat@proofof@#2}}%
	}%
	\fi
	\@namedef{pat@proofof@#2}{#1}%
}

\declaretheorem[style=plain,name=Theorem, numberwithin=section]{theorem}

\newtheorem{corollary}{Corollary}[section]
\newtheorem{lemma}[theorem]{Lemma}
\newtheorem{assumption}{Assumption}

\newtheorem{remark}{Remark}

\usepackage{pgfplotstable}

\title{Exact Bias Correction for Linear Adjustment of Randomized Controlled Trials}
\author{Haoge Chang, Joel A. Middleton, and P. M. Aronow\footnote{We thank Donald Andrews, Winston Lin, Cyrus Samii and Jasjeet Sekhon for helpful comments and discussions.}}
\begin{document}

\maketitle
\begin{abstract}

In an influential critique of empirical practice, Freedman \cite{freedman2008A,freedman2008B} showed that the linear regression estimator was biased for the analysis of randomized controlled trials under the randomization model. Under Freedman's assumptions, we derive exact closed-form bias corrections for the linear regression estimator with and without treatment-by-covariate interactions. We show that the limiting distribution of the bias corrected estimator is identical to the uncorrected estimator, implying that the asymptotic gains from adjustment can be attained without introducing any risk of bias. Taken together with results from Lin \cite{lin2013agnostic}, our results show that Freedman's theoretical arguments against the use of regression adjustment can be completely resolved with minor modifications to practice. 

\end{abstract}

\section{Introduction}

Randomized Controlled Trials (RCTs) are popular in empirical economics \cite{angrist2008mostly,duflo2007using,glennerster2017practicalities,list2011field}. When estimating average treatment effects, adjustment for pretreatment covariates with linear regression is a commonly recommended practice because it can reduce the variability of estimates. However, adjusting for covariates remains somewhat controversial, in large part because of an influential critique from David Freedman \cite{freedman2008A,freedman2008B}. 

Freedman argued that randomization does not justify the use of linear regression for completely randomized experiments. Freedman's theoretical arguments relied on three results proven under the randomization-based \cite{splawa1990application,imbens2015causal}  inferential paradigm: 
\begin{enumerate}
\item asymptotically, the linear regression estimator can be inefficient relative to the unadjusted (difference-in-means) estimator if the design is imbalanced; \item the classical standard error for linear regression is inconsistent; 
\item the regression estimator has an $O_p(n^{-1})$ bias term.
\end{enumerate}
\noindent Freedman's arguments were influential among scholars across multiple disciplines (e.g., \cite{dunning2012natural},\cite{BenRecht}). Freedman's third argument garnered particular attention among social scientists. Notably, \cite{deaton2018understanding}'s critique of randomization in empirical economics argued that the bias introduced by regression undermines the gold standard argument for RCTs.

Scholars have worked to address Freedman's critiques and to understand the extent that they can and do matter for empirical practice in empirical economics. Using Freedman’s own framework, \cite{lin2013agnostic} showed that arguments 1 and 2 were resolved by  small modifications to practice. Freedman's efficiency result may be addressed through simple modifications to the regression specification, namely including treatment by covariate interactions \cite{blinder1973wage,oaxaca1973male}. Then it can be shown that the adjusted estimator is never less asymptotically efficient than the unadjusted estimator.  Regarding argument 2, Lin proves that robust standard errors (\cite{white1980heteroskedasticity,Hube67,eicker1967limit}, see also \cite{samii2012equivalencies}) are asymptotically conservative in Freedman's setting, guaranteeing the validity of large-sample inference. On argument 3, \cite{lin2013agnostic} (see also \cite{lin2013essays}) notes that the leading term of the bias is in fact estimable and can be shown to be small in a real-world empirical example. However, the small-sample bias of the regression estimator was not yet fully resolved. 

Since \cite{lin2013agnostic}, there have been notable papers that have proposed unbiased regression-type estimators for experimental data. \cite{miratrix2013adjusting} demonstrate that if the regression model is fully saturated (see also \cite{athey2017econometrics} and \cite{imbens2010better}), then the associated effect estimate is unbiased conditional on the event that treatment is not collinear with any covariate stratum. This approach cannot generally be used without coarsening continuous covariates. \cite{aronow2013class} proposes the use of auxiliary data, demonstrating that the suitable use of hold-out samples ensures the finite-sample unbiasedness of the associated regression estimator, but the paper does not consider efficiency properties. More recently, \cite{wu2018loop} extended \cite{aronow2013class} to propose an innovative but computationally expensive split-sample approach for completely randomized experiments.

The primary contribution of this paper is to resolve Freedman's third theoretical argument by proposing finite-sample-exact, closed-form bias corrections without adding any new assumptions. 
 Our idea builds on \cite{lin2013agnostic}’s proposal to estimate the leading term of the bias, but further develops a novel finite-sample exact bias correction encompassing all higher-order terms  \cite{freedman2008A}. We derive these bias corrections for both the noninteracted and interacted linear regression estimators. We prove that the estimators have the same limiting distributions as the non-bias-adjusted estimators, implying that they could replace existing estimators in instances where bias is a prevailing concern (e.g., trials that may be aggregated in meta-analysis). We further provide a numerical illustration demonstrating these properties.\\
\indent Finally, we remind the readers that the practice of debiasing estimators is not uncontroversial.  \cite{tibshirani1993introduction}\footnote{We thank Winston Lin for suggesting this reference.} has warned that the bias correction could be dangerous in practice due to its high variability\footnote{As will be shown in the simulations, when the performance is measured by the Root Mean Squared Error (RMSE), there is no clear dominance among the estimators: in some cases the RMSEs of the debiased estimators  are strictly smaller than those of other estimators, and in other cases larger.}. In real-world decision making processes, people may express different preferences for different statistical properties (i.e. unbiasedness or low Mean Squared Error)\footnote{See \cite{wu2018loop} for an anecdotal example of a policy-maker favoring unbiasedness.}. Our results shall imply that in large samples the additional variation caused by the bias correction is negligible, but for small samples, in some cases, we find it important to account for the sampling variability of the additional terms. To address this problem, we propose a simple modification to the standard error estimation procedure. Such modification, based on recomputing OLS residuals using the debiased estimators, is shown to work well on our simulated datasets. We make recommendations for practice in the Simulation section.\\
\indent The organization of the paper is as follows: Section 2 includes the model setup and assumptions; Section 3 considers a characterization of bias terms of the OLS estimators; Section 4 proposes the bias corrections; Section 5 includes simulation results with both simulated datasets and a real world dataset. In the appendix one can find the proofs for theorems in Section 3 and Section 4, and more simulation results. 

\section{Setting, Assumptions and Notations}
\indent We follow the setting of \cite{freedman2008A} and \cite{lin2013agnostic}, which assume a Neyman \cite{splawa1990application} model with covariates. 
There are $n$ subjects indexed by $i=1,...,n$. For each subject we observe an outcome $Y_i$ and a column vector of covariates $\textbf{z}_i=(z_{i1},z_{i2},...,z_{iK})\in\mathbb{R}^K$. The dimension of the covariates, $K$, does not change with the sample size.\\
\indent Each subject has two potential outcomes $a_i$ and $b_i$ (cf., the stable-unit-treatment-value-assumption \cite{RUBIN1990279}). We observe $Y_i=a_i$ if $i$ is chosen for treatment arm $A$ (treated group) and $Y_i=b_i$ if i is chosen for arm $B$ (control group). Let $T_i$ be the dummy variable for treatment arm $A$. Thus the observed outcome for $i$ is $Y_i=a_iT_i+b_i(1-T_i)$. \\
\indent The experiment is assumed to be completely randomized: $n_A$ out of $n$ subjects are 
randomly assigned to arm $A$ and the remaining $n_B=n-n_A$ subjects to arm $B$. Random assignment is the only source of randomness in the model. We do not assume a superpopulation: the $n$ subjects are the population of interest.\\
\indent We introduce some notation. let $n$ be the population size, $n_A$ and $n_B$ be the number of subjects in treatment arms $A$ and $B$, respectively. Let $[A]=\{i\mid T_i=1\}$ denote the set of individuals who are chosen for arm $A$ and similarly $[B]=\{i\mid T_i=0\}$.    Let $\Bar{x}=\frac{1}{N}\sum_{i=1}^N x_i$, $\Bar{x}_A=\frac{1}{n_A}\sum_{i\in[A]} x_i$ and $\Bar{x}_B=\frac{1}{n_B}\sum_{i\in[B]} x_i$ denote the population average, group $A$ average, and group $B$ average, respectively, of possibly a vector-valued variable $x$. The average treatment effect (ATE) can be written in this notation as:
\begin{equation*}
    \bar{a}-\bar{b}
\end{equation*}
and the difference-in-means estimator:
\begin{equation*}
    \aA - \bB
\end{equation*}
Simiarly we can write $\frac{1}{n}\sum_{i=1}^n \mathbf{z}_i\mathbf{z}_i'=\overline{\mathbf{z}\mathbf{z}'}$ for $\mathbf{z}_i\in\mathbb{R}^K$ and $\frac{1}{n}\sum_{i=1}^na_i\mathbf{z}_i=\az$ for $a_i\in\mathbb{R}$ and $\mathbb{z}\in\mathbb{R}^K$.\\
\indent We make the following assumptions throughout the paper, which are standard in the literature. (cf., \cite{freedman2008B,freedman2008A,lin2013agnostic}).
\begin{assumption}[Bounded fourth moments]
For all $n=1,2...$, and $x_i\in\{a_i,b_i,z_{i1},...,z_{iK}\}$,
\begin{align*}
\frac{1}{n}\sum_{i=1}^n x_i^4<L<\infty 
\end{align*}
where $L$ is a finite constant.
\end{assumption}
\begin{assumption}[Convergence of first and second moments]
For $x_i=[a_i,b_i,\z_i']\in\mathbb{R}^{2+K}$,
\begin{align*}
    \frac{1}{n}\sum_{i=1}^n x_ix_i' \to \mathbf{M}
\end{align*}
where $\mathbf{M}$ is a positive definite matrix with finite entries. Moreoever, $\frac{1}{n}\sum_{i=1}^n\z_i\z_i'$ converges to an invertible matrix.
\end{assumption}
\begin{assumption}[Group Sizes]
Let $p_{A,n}=\frac{n_A}{n}$ and $p_{B,n}=\frac{n-n_A}{n}$, the inclusion probability into the treatment arm A and arm B, respectively. We assume $p_{A,n}>0$ and $p_{B,n}>0$ for all $n$, and 
\begin{align*}
p_{A,n} \to p_A>0, \text{ as } n\to\infty
\end{align*}
\begin{align*}
p_{B,n} \to 1-p_A>0, \text{ as } n\to\infty
\end{align*}
\end{assumption}
\begin{assumption}[Centering]
$\bar{\z} = 0$
\end{assumption}
All assumptions are employed regularly in the literature. They are used to derive consistency and asymptotic normality for the estimator. Assumption 3 requires each arm receives a nontrivial fraction of subjects over the asymptotic sequence of the models. Assumption 4 is without loss of generality: in practice, researchers can just demean each covariate and apply our method.\\
\indent We remind the readers of the definitions of our two OLS regression adjusted ATE estimators. The first estimator comes from a noninteracted OLS regression where one regresses observed outcome $Y$ on the treatment indicator $T$ and demeaned pretreatment covariates $Z$. The coefficient estimate for $T$ is the \textit{noninteracted} OLS regression adjusted ATE estimator. The second estimator comes from a fully interacted OLS regression where one regresses observed outcome $Y$ on the treatment indicator $T$, demeaned pretreatment covariates $Z$, and interaction terms of the treatment indicators and demeaned pretreatment covariates. The coefficient estimate for $T$ is the \textit{interacted} OLS regression adjusted ATE estimator.\\
\indent Finally, we prepare some notation for the sections below. Let $a^*_i$ and $b^*_i$ be the centered potential outcomes, namely, $a^*_i=a_i-\bar{a}$ and $b^*_i=b_i-\bar{b}$. Let $\widehat{D}=\zz - \pA \zA \zA'- \pB \zB  \zB'$, $\N=\pA (\azA-\aA\zA) +\pB (\bzB-\bB \zB)$, $D=\zz$ and $N=\pA\az + \pB\bz$. With this notation the regression coefficients estimators of the pretreatment covariates in the noninteracted case can be written as $\Q={\D}^{-1}\N$, and the population coefficients $Q=D^{-1}N$. Denote the (rescaled) leverage of $i$th data point as
$h_i=z_i'D^{-1}z_i$.\\
\indent Further define $\D_A= \zzA - \zA\zA'$, $\N_A = \azA-\aA\zA$, $\D_B= \zzB - \zB\zB'$, $\N_B = \bzB-\bB\zB$, $N_A=\az$ and $N_B=\bz$. Thus the regression coefficients estimator of the pretreatment covariates in the interactive case can be written as $\Q_A = \D^{-1}_A\N_A$ and $\hat{Q}_B = \D^{-1}_B\N_B$. Their population counterparts are $Q_A=D^{-1}N_A$ and $Q_B=D^{-1}N_B$. 
\section{Bias Characterization}
As shown in \cite{lin2013agnostic}, the OLS regression adjusted ATE estimator can be written as:
\begin{align*}
\widehat{ATE}_{NI}= \aA  - \bB - \{(\zA- \overline{\z})' \Q -(\zB-\overline{\z})'  \Q \}
\end{align*}
for the noninteracted case and 
\begin{align*}
  \widehat{ATE}_{I} & = \aA  - \bB -\{ (\zA-\overline{\z})'\Q_A - (\overline{\z}_B -\overline{\z})'\Q_B\}
\end{align*}
for the interacted case, where $\Q$, $\Q_A$ and $\Q_B$ are the OLS coefficients in front of the covariates.\\

\indent A characterization of the bias terms is provided in this section. Note that both the noninteracted and interacted estimators can be written as sums of the difference-in-means estimator and a regression adjustment using group means and OLS coefficients. The bias comes from the regression adjustment terms, in particular from estimating the regression coefficients on the covariates. We shall characterize the bias terms of the coefficient estimator first. We start from the noninteracted case. Note from here on we shall assume for simplicity that all design matrices (i.e. $\hat{D}$ and $D$) are invertible. In case of noninvertible design matrices, our debiased procedure will still work after choosing a particular generalized inverse, and compute the ATE estimators according to the formulae above.  
\begin{theorem}
	The OLS coefficient vector for covariates can be written as:
	\begin{align*}
		\Q =&  Q+ \nu_1 +\nu_2 +\nu_3
	\end{align*}
	with 
	\begin{align*}
		\nu_1 &= D^{-1}\left(\pA \left(\aszA-\asz\right) +\pB \left(\bszB-\bsz\right) \right)= O_p(n^{-\frac{1}{2}}), \\
	\nu_2 &=  \left(\D^{-1} - D^{-1}\right)\N = O_p(n^{ -1}), \\
   \nu_3 &=  - D^{-1} \left(\pA \asA \zA + \pB \bsB \zB\right) =  O_p(n^{-1}). 
	\end{align*}
\end{theorem}
From the coefficient decomposition one can directly characterize the bias term of the ATE estimator. Note that the bias terms are of order $O_p(n^{-1})$.
\begin{corollary}
	The bias of the $\widehat{ATE}_{NI}$ estimator is:
	\begin{align*}
	\E[ \widehat{ATE}_{NI} -ATE ] =& \E \left[ (\zB-\zA)(\nu_1+\nu_2+\nu_3) \right].
	\end{align*}    
	Moreover, $(\zB-\zA)(\nu_1+\nu_2+\nu_3)=O_p(n^{-\frac{1}{2}})O_p(n^{-\frac{1}{2}})=O_p(n^{-1})$
\end{corollary}
\indent Following the same steps as we did for the noninteracted estimator, we are able to derive analogous results for the interacted estimator.
\begin{theorem}
The OLS coefficient vectors for covariates can be written as
\begin{align*}
\hat{Q}_A = Q_A + \nu_{1A} + \nu_{2A}  \\
\hat{Q}_B = Q_B + \nu_{1B} + \nu_{2B}
\end{align*}
with 
\begin{align*}
    \nu_{1A} &= (\hat{D}^{-1}_A - D^{-1}) \hat{N}_A=O_p(n^{-\frac{1}{2}})\\
     \nu_{2A} &= D^{-1} (\hat{N}_A- N_A) = O_p(n^{-\frac{1}{2}}); \\
    \nu_{1B} &= (\hat{D}^{-1}_B - D^{-1}) \hat{N}_B=O_p(n^{-\frac{1}{2}}); \\
    \nu_{2B} & =D^{-1} (\hat{N}_B- N_B) = O_p(n^{-\frac{1}{2}})
\end{align*}
\end{theorem}
\begin{corollary}
    The bias of the $\widehat{ATE}_{I}$ estimator is:
\begin{align*}
    E[\widehat{ATE}_I-ATE] & = E[\bar{z}_B(\nu_{1B} + \nu_{2B})] - E[\bar{z}_A(\nu_{1A} + \nu_{2A})] 
\end{align*}
	Moreover, $\bar{z}_B(\nu_{1B} + \nu_{2B})=O_p(n^{-\frac{1}{2}})O_p(n^{-\frac{1}{2}})=O_p(n^{-1})$ and $\bar{z}_A(\nu_{1A} + \nu_{2A})=O_p(n^{-\frac{1}{2}})O_p(n^{-\frac{1}{2}})=O_p(n^{-1})$.
\end{corollary}
 Note that this result implies that the bias terms of the interactive ATE estimator are also of order $O_p(n^{-1})$.
 
\section{Bias Corrections for Regression Components}

Having established the decomposition, we now derive estimators of each bias term for use as bias corrections. We show that these bias estimates are (i) exactly unbiased and (ii) have estimation error of $O_p(n^{-1})$. It follows that use of this bias correction with an adjusted estimator yields a finite-sample unbiased estimator with the limit distribution of the adjusted estimator. We remind the readers that $h_i=z_i'D^{-1}z_i$, the (rescaled) leverage of $i$th data point as defined in Section 2.


We again begin with the noninteracted case.
\begin{theorem}
    An unbiased estimator for the bias in the noninteracted case is:
\begin{align*}
  \widehat{Bias}_{NI}= &\frac{1}{n}\frac{n_B}{n_B-1}\left( \hbB-\overline{h}_B \bB  \right)- \frac{1}{n}\frac{n_A}{n_A-1} \left(\haA -\overline{h}_A \aA\right) 
		+ (\zB-\zA)' \left(\D^{-1} - D^{-1}\right)\N +\\
  &\frac{C_{A,NI}}{n_A}\sum_{i\in [A]}(z_i-\zA)'D^{-1}(z_i-\zA)(a_i-\aA)-
  \frac{C_{B,NI}}{n_B}\sum_{i\in [B]}(z_i-\zB)'D^{-1}(z_i-\zB)(b_i-\bB)
\end{align*}   
$C_{A,NI}$ and $C_{B,NI}$ are two constants depending on $n$, $n_A$ and $n_B$. Their exact formulas are given in the appendix.
Moreover $\widehat{Bias}_{NI}=O_p(n^{-1})$
\end{theorem}
\begin{corollary}
    The following estimator is unbiased for estimating the ATE:
    $$\widehat{ATE}_{NI,Debias}=\widehat{ATE}_{NI} -  \widehat{Bias}_{NI}$$
\end{corollary}


The results are derived analogously in the interacted case.

\begin{theorem}
    An unbiased estimator for the bias in the interacted case is:
\begin{align*}
  \widehat{Bias}_I= &\frac{1}{n} \frac{n_A}{n_B-1} \left( \hbB-\overline{h}_B \bB  \right) + \bar{z}_B'(\hat{D}_B^{-1}-D^{-1})\hat{N}_B - \frac{C_{B,I}}{n_B} \sum_{i\in [B]}(z_i-\zB)'D^{-1}(z_i-\zB)(b_i-\bB)- \\
   & \frac{1}{n} \frac{n_B}{n_A-1} \left(\haA -\overline{h}_A \aA\right)  - \bar{z}_A'(\hat{D}_A^{-1}-D^{-1})\hat{N}_A + \frac{C_{A,I}}{n_A}\sum_{i\in [A]}(z_i-\zA)'D^{-1}(z_i-\zA)(a_i-\aA). 
\end{align*}   
$C_{A,I}$ and $C_{B,I}$ are two constants depending on $n$, $n_A$ and $n_B$. Their exact formulas are given in the appendix.
Moreover $\widehat{Bias}_I=O_p(n^{-1})$. 
\end{theorem}
\begin{corollary}
    The following estimator is unbiased for estimating the ATE:
    $$\widehat{ATE}_{I,Debias}=\widehat{ATE}_{I} -  \widehat{Bias}_I$$
\end{corollary}
\begin{remark}
    Note that both adjustments in Theorem 4.1 and Theorem 4.2 are of order $O_p(n^{-1})$. Thus $\sqrt{n}(\widehat{ATE}_{NI,Debias}-\widehat{ATE}_{NI})=o_p(1)$ and $\sqrt{n}(\widehat{ATE}_{I,Debias}-\widehat{ATE}_{I})=o_p(1)$. The debiased estimators have the same limiting distributions as the original estimators.
\end{remark}
\begin{remark}
We briefly remark on why it is possible to design unbiased adjusted estimators in closed-form. Examining at the expressions in Section 3, although all biases are nonlinear, only $v_2$, $v_{1A}$ and $v_{1B}$ have an infinite order Taylor expansion, but these terms can be expressed purely in terms of the observable data. All other terms are in expectation functions of moments that can be unbiasedly estimated.
\end{remark}
\section{Simulations}
In this section we apply our estimators on several datasets. 

We briefly comment on variance estimation and confidence interval construction. We showed in Section 4 that our debiased estimators have the same asymptotic distributions as those of the OLS estimators. This implies that for large samples we can recenter the OLS confidence intervals with our debiased estimators and expect the same coverage probabilities. For small samples, however, we find it important to account for the sampling variability of the additional terms. Indeed, in one of the simulations below, we found that a naive recentering procedure may lead to severe undercoverage. To address this problem, we propose a simple procedure shown to work well on our simulated datasets.\footnote{Another way is to directly estimate the variances of the additional terms, but this may be cumbersome to do.} In this procedure, one first runs the OLS regression and computes the debiased estimator. Then one replaces the OLS treatment coefficient with the debiased estimated coefficients and recomputes the OLS residuals, keeping all other coefficients the same. Finally, one computes the variances and constructs confidence intervals for the debiased estimator using the same formula as for the OLS estimators. In the simulations below, such procedures will be denoted by BC, which stands for bias-corrected. In Appendix B, one can find a more detailed comparison of this new procedure with standard ones. In practice, we recommend researchers to use our debiased estimators with this procedure, the BC-HC2 heteroskedasticity-robust variance estimator with a Satterthwaite adjustment for inference\footnote{For a discussion of the Satterthwaite adjustment, see \cite{satterthwaite1946approximate},\cite{bell2002bias}, \cite{lin2013agnostic} and \cite{imbens2010better}.}. 

\subsection{Simulated Datasets}
In this section we compare the performance of our debiased estimators with that of standard estimators using simulated datasets. We show the results of two simulation schemes here\footnote{In Appendix B, readers can find results for two more simulation schemes as well as some graphical information of the data generating processes.}.  In each scheme, we first generate two dimensional covariates that are the quantiles of some prespecified distributions\footnote{For example, in scheme 1, with a sample of $N$ units, the covariates of $i$th unit are the $\frac{i}{N+1}$ quantiles of a Beta(0.5,0.5) distribution and a Triangle(0,1) distribution.}. We then compute the studentized leverage ratios and use those to impute potential outcomes. We consider three ways to impute potential outcomes. For all cases the average treatment effect is equal to 0. The experiment is a completely randomized experiment with 24 units and an inclusion probability of $\frac{1}{3}$ for the treatment arm. Table~\ref{Table1} includes simulation details. Note that these schemes are designed specifically such that the finite sample bias is relatively large. \\
\indent Table 2 and 3 report the simulation results for the two schemes. Our debiased estimators are exactly unbiased as expected. In terms of the root mean squared error (RMSE) the picture is less clear. There are cases where the debiased estimators dominated others (DGP 1.1 and DGP 1.3), and cases where the unadjusted estimator is the best (DGP 1.2 and DGP 2.2)\footnote{This is the artifact of the DGPs. Recall the variance formula for the difference-in-means estimator, for example from \cite{imbens2015causal}.}. Note that DGP 1.3 and DGP 2.3 are constant effect models, in which the noninteracted OLS estimators are first order unbiased. However in DGP 1.3 we still observe a small, higher-order bias.     \\
\indent In terms of confidence interval coverages, observe in DGP 2.1, 2.2 and 2.3 that the original recentering intervals exhibit significant undercoverages. However, the procedure based on recomputing the OLS residuals with Satterthwaite adjustments works reasonably well. There is only one case DGP2.3, Non-Int, where the coverage is not very satisfactory. As shown in Table 6 and 8 in the appendix, the BC procedures do not add to the median confidence interval length significantly, although it tends to add to the average confidence interval length. The Satterthwaite could also add to the median (mean) confidence interval length: it typically increases the confidence interval length by at most 10 to 20 percent, compared with the Student-t adjustment. 
\begin{table}[H] 
\centering 
\begin{threeparttable}
\begin{tabular}{l c c c c c c } 
\toprule 
\cmidrule(l){2-6} 
 & $X_1(i)$ & $X_2(i)$ & $Y_0(i)$ & $Y_1(i)$ & \#treated & ATE \\ 
\midrule 
\textbf{Scheme 1, N=24}& & & & & &  \\
DGP1.1 &\multirow{3}{*}{Beta(0.5,0.5)}& \multirow{3}{*}{Tri(0,1)} & 0 & 2$h_i$ & \multirow{3}{*}{8} &  \multirow{3}{*}{0}\\
DGP1.2 & & & -$h_i$ & $h_i$ &&\\
DGP1.3& & & $h_i$ & $h_i$ && \\
& & & & \\
\textbf{Scheme 2, N=24}& & & &  \\
DGP2.1 &\multirow{3}{*}{Beta(2,5)}& \multirow{3}{*}{Norm(0,1)} & 0 & 2$h_i$ & \multirow{3}{*}{8} &  \multirow{3}{*}{0} \\
DGP2.2& & & -$h_i$ & $h_i$&&\\
DGP2.3 & & & $h_i$ & $h_i$&&\\
\bottomrule
\end{tabular}
 \begin{tablenotes}
      \small
\item Note: DGPs for Simulations. Beta($\alpha$,$\beta$) are the beta distributions with shape parameters $\alpha$ and $\beta$. Tri(0,1) is the symmetric triangular distribution on the unit interval. Norm(0,1) is the standard normal distribution.  $h_i$ is the studentized leverage ratio for the $i$th unit. It is computed as $h_i=\frac{v_i-\bar{v}}{\sigma_v}$, where $v_i=x_i'(\sum_{i=1}^nx_ix_i')^{-1}x_i$, $\bar{v}=\frac{1}{n}\sum_{i=1}^nv_i$ and $\sigma^2_v=\frac{1}{n-1}\sum_{i=1}^N(v_i-\bar{v}^2)$. Note DGP1.1, DGP1.2, DGP2.1 and DGP2.2 are variable effects models, and DGP 1.3 and DGP 2.3 are constant effects models.
\end{tablenotes}
\label{Table1}
\end{threeparttable}
\end{table}

\begin{table}[H]
\centering
\caption{Simulation Results for Scheme 1}    
\begin{threeparttable}
\begin{tabular}{lccccc}
\toprule 
& \multicolumn{5}{c}{\textbf{ATE Estimators}} \\ 
\cmidrule(l){2-6} 
 & Unadjusted & \multicolumn{2}{c}{OLS} & \multicolumn{2}{c}{Debiased} \\ 
  & & Non-Int. & Interacted & Non-Int. & Interacted \\ 
\midrule 
\textbf{DGP1.1, N =24}& & & & & \\
Bias & -0.000 & -0.044 & -0.171 & -0.000 & -0.000 \\ 
  SD & 0.577 & 0.569 & 0.734 & 0.558 & 0.570 \\ 
  RMSE & 0.577 & 0.571 & 0.754 & 0.558 & 0.570 \\ 
  \makecell[l]{CI Coverage  (HC2, Student-t)} & 0.961 & 0.957 & 0.919 &  0.960 & 0.953 \\ 
    \makecell[l]{CI Coverage  (HC2, Satterthwaite)} & 0.965 & 0.964 & 0.949 &  0.966 & 0.970 \\ 
  \makecell[l]{CI Coverage  (BC-HC2, Student-t)} &  &  && 0.961 & 0.957 \\  
  \makecell[l]{CI Coverage  (BC-HC2, Satterthwaite)} &  &  && 0.967 & 0.973 \\  
&&&&&\\
 \textbf{DGP1.2, N =24}& & & & & \\
 Bias & -0.000 & -0.046 & -0.097 & -0.000 & -0.000 \\ 
  SD & 0.144 & 0.220 & 0.275 & 0.205 & 0.182 \\ 
  RMSE & 0.144 & 0.225 & 0.292 & 0.205 & 0.182 \\ 
 \makecell[l]{CI Coverage  (HC2, Student-t)} & 1.000 & 0.999 & 0.982 & 1.000 & 0.999 \\ 
     \makecell[l]{CI Coverage  (HC2, Satterthwaite)} & 1.000 & 1.000 & 0.994 &  1.000 &1.000 \\ 
   \makecell[l]{CI Coverage  (BC-HC2, Student-t)} &  &  && 1.000 & 1.000 \\  
 \makecell[l]{CI Coverage (BC-HC2, Satterthwaite)} &  &  &  & 1.000 & 1.000 \\  
&&&&&\\
 \textbf{DGP1.3, N =24}& & & & & \\
  Bias & -0.000 & 0.002 & -0.074 & 0.000 & 0.000 \\ 
  SD & 0.433 & 0.417 & 0.483 & 0.400 & 0.408 \\ 
  RMSE & 0.433 & 0.417 & 0.489 & 0.400 & 0.408 \\ 
 \makecell[l]{CI Coverage  (HC2, Student-t)} & 0.940 & 0.938 & 0.916 & 0.946 & 0.948 \\ 
  \makecell[l]{CI Coverage  (HC2, Satterthwaite)} & 0.947 & 0.949 & 0.950 & 0.956 & 0.970 \\ 
    \makecell[l]{CI Coverage  (BC-HC2, Student-t)} &  &  && 0.947 & 0.950 \\  
 \makecell[l]{CI Coverage  (BC-HC2, Satterthwaite)} &  &  &  &0.956 & 0.971 \\ 
   \hline  
\end{tabular}
 \begin{tablenotes}
      \small
\item Note: CI Coverage measures the empirical coverage rates of nominal 95 percent confidence intervals. The unit of CI coverage is $\times$100 percentage points. CI Coverage (HC2, Student-t) is calculated using the original OLS residuals. CI Coverage (BC-HC2, Student-t) and CI Coverage (BC-HC2, Satterthwaite) are calculated using the recomputed OLS residuals. Both Student-t and Satterthwaite adjustments are calculated using the R-package \texttt{clubSandwich}\cite{clubSandwich}. The number of simulation is ${24 \choose 8}$.
\end{tablenotes}
\end{threeparttable}
\end{table}
\begin{table}[H]
\centering
\caption{Simulation Results for Scheme 2}    
\begin{threeparttable}
\begin{tabular}{lccccc}
\toprule 
& \multicolumn{5}{c}{\textbf{ATE Estimators}} \\ 
\cmidrule(l){2-6} 
 & Unadjusted & \multicolumn{2}{c}{OLS} & \multicolumn{2}{c}{Debiased} \\ 
  & & Non-Int. & Interacted & Non-Int. & Interacted \\ 
\midrule 
\textbf{DGP2.1, N =24}& & & & & \\
Bias & 0.000 & -0.237 & 0.028 & 0.000 & 0.000 \\ 
  SD & 0.577 & 0.344 & 0.283 & 0.459 & 0.439 \\ 
  RMSE & 0.577 & 0.418 & 0.284 & 0.459 & 0.439 \\ 
 \makecell[l]{CI Coverage \\ (HC2, Student-t)} & 0.910 & 0.913 & 0.757 & 0.923 & 0.470 \\ 
 \makecell[l]{CI Coverage \\ (HC2, Satterthwaite)} & 0.915& 0.920 & 0.837 & 0.928 & 0.548 \\ 
  \makecell[l]{CI Coverage \\ (BC-HC2, Student-t)} & &  &  & 0.923 & 0.876 \\ 
 \makecell[l]{CI Coverage \\ (BC-HC2, Satterthwaite)} & &  &  & 0.928 & 0.930 \\ 
&&&&&\\
\textbf{DGP2.2, N =24}& & & & & \\
Bias & 0.000 & -0.237 & 0.015 & 0.000 & 0.000 \\ 
  SD & 0.144 & 0.326 & 0.132 & 0.314 & 0.225 \\ 
  RMSE & 0.144 & 0.403 & 0.133 & 0.314 & 0.225 \\ 
  Coverage\_HC2 & 1.000 & 0.929 & 0.933 & 0.969 & 0.724 \\ 
 \makecell[l]{CI Coverage  (HC2, Student-t)}  & 1.00 & 0.93 & 0.93 & 0.967 & 0.614 \\ 
  \makecell[l]{CI Coverage  (HC2, Satterthwaite)}  & 1.00 & 0.935 & 0.991 & 0.969 & 0.724 \\ 
  \makecell[l]{CI Coverage  (BC-HC2, Student-t)} & &  &  & 0.965 & 0.967 \\ 
 \makecell[l]{CI Coverage  (BC-HC2, Satterthwaite)} &  & &  & 0.968 & 0.996 \\ 
&&&&&\\
 \textbf{DGP2.3, N =24}& & & & & \\
 Bias & 0.000 & 0.000 & 0.013 & 0.000 & 0.000 \\ 
  SD & 0.433 & 0.097 & 0.163 & 0.195 & 0.239 \\ 
  RMSE & 0.433 & 0.097 & 0.164 & 0.195 & 0.239 \\
 \makecell[l]{CI Coverage (HC2, Student-t)}  & 0.93 & 0.97 & 0.85 & 0.654 & 0.570 \\ 
 \makecell[l]{CI Coverage  (HC2, Satterthwaite)}  & 0.942 & 0.983  &0.947 & 0.683 & 0.678 \\ 
  \makecell[l]{CI Coverage  (BC-HC2, Student-t)} & &  &  & 0.809 & 0.896 \\ 
 \makecell[l]{CI Coverage  (BC-HC2, Satterthwaite)}  &  &  & & 0.850 & 0.944 \\ 
   \hline  
\end{tabular}

 \begin{tablenotes}
      \small
\item Note: CI Coverage measures the empirical coverage rates of nominal 95 percent confidence intervals. The unit of CI coverage is $\times$100 percentage points. CI Coverage (HC2, Student-t) is calculated using the original OLS residuals. CI Coverage (BC-HC2, Student-t) and CI Coverage (BC-HC2, Satterthwaite) are calculated using the recomputed OLS residuals. Both Student-t and Satterthwaite adjustments are calculated using the R-package \texttt{clubSandwich}\cite{clubSandwich}. The number of simulation is ${24 \choose 8}$.
\end{tablenotes}
\end{threeparttable}
\end{table}

\subsection{Real Dataset}

In this section we compare the performance of debiased estimators with that of standard OLS estimators on a real world dataset. We precisely follow \cite{lin2013agnostic}'s simulation setting. We generate our simulation from the experimental example of \cite{angrist2009incentives} by simulating random assignments under the maintained hypothesis of no treatment effect. Because this setting assumes no effects, bias is expected to be negligible: \cite{freedman2008A} notes that the leading term of the bias is greatest when treatment effects are heterogeneous. The simulation is thus not primarily meant to investigate bias, but rather the precision and coverage consequences of the use of our bias corrections in a real-world setting.\footnote{We thank Winston Lin for sharing the replication files.} \\
\indent  \cite{angrist2009incentives} sought to measure the effects of support services and financial incentives on college students' academic achievement. The experiment randomly assigns eligible first-year undergraduate students into four groups. One treatment group was offered both support services and financial incentives. A second group was offered only support services and a third group only financial incentives. The control group was eligible only for standard university support services. As in \cite{lin2013agnostic}, we only use the data for men in the services-and-incentives ($N=58$) and service-only ($N=99$) groups. The simulation datasets are generated assuming the treatment has no effect on any students. We replicate the experiments $10^7$ times, and each time randomly assign $58$ students to the services-and-incentives group and $99$ to the service-only group. The regression estimators estimate the treatment effects adjusting for high-school GPAs. The standard errors of the OLS estimators are estimated using the standard sandwich formulas.\\
\indent Table 1 reports the simulation results from $10^7$ simulations. The first and second rows of the table show the means and standard deviations of the five estimators. All estimators are approximately unbiased after rounding and the variances of the debiased estimators are no larger than the standard estimators. The third row shows that all confidence intervals cover the true value of the ATE with approximately 95 percent probability. The fourth row reports the average length of the confidence intervals. On average, the intervals of regression-adjusted estimators are slightly narrower than that of the unadjusted estimator. (The width of the confidence intervals for the debiased estimators are mechanically identical to those for the standard estimators due to being constructed using the same SE estimators.)

\begin{table}[H]
\centering
\caption{Simulation of Angrist, Lang, Oreopolous (2009) experiment with zero treatment effects ($10^7$ replications).}
\begin{threeparttable}
\begin{tabular}{l c c c c c}
\toprule
& \multicolumn{5}{c}{\textbf{ATE Estimators}} \\ 
\cmidrule(l){2-6} 
 & Unadjusted & \multicolumn{2}{c}{OLS} & \multicolumn{2}{c}{Debiased} \\ 
  & & Non-Int. & Interacted & Non-Int. & Interacted \\ 
\midrule 
Bias & 0.000 & 0.000 & 0.000 & 0.000 & 0.000\\ 
SD & 0.159 & 0.150 & 0.147 & 0.147 & 0.147\\ 
MSE & 0.159 & 0.150 & 0.147 & 0.147 & 0.147\\ 
\makecell[l]{CI Coverage  (HC2, Student-t)}& 0.949 & 0.949 & 0.949 & 0.949 & 0.949\\ 
\makecell[l]{CI Coverage  (HC2, Satterthwaite)}& 0.949 & 0.949 & 0.949 & 0.950 & 0.950\\ 
  \makecell[l]{CI Coverage  (BC-HC2, Student-t)} &  & &  & 0.949 & 0.949 \\
 \makecell[l]{CI Coverage  (BC-HC2, Satterthwaite)} &  &&  & 0.950 & 0.950 \\
\bottomrule
\end{tabular}
 \begin{tablenotes}
      \small
\item Note: CI Coverage measures the empirical coverage rates of nominal 95 percent confidence intervals. The unit of CI coverage is $\times$100 percentage points. CI Coverage (HC2, Student-t) is calculated using the original OLS residuals. CI Coverage (BC-HC2, Student-t) and CI Coverage (BC-HC2, Satterthwaite) are calculated using the recomputed OLS residuals. Both Student-t and Satterthwaite adjustments are calculated using the R-package \texttt{clubSandwich}\cite{clubSandwich}.
\end{tablenotes}
\end{threeparttable}
\end{table}
Together, these results demonstrate that, in a real-world setting, our bias corrections can be effectively introduced without appreciably compromising the precision or coverage properties of regression adjusted estimators.


\pagebreak[4]
\pagebreak[4]
\appendix 

\section{Proofs}
\subsection{Constants}
We define following five constants:
\begin{align*}
    N_{AAA} &= \frac{n}{n_A^3}(\frac{n_A}{n}-\frac{3n_A(n_A-1)}{n(n-1)}+\frac{2n_A(n_A-1)(n_A-2)}{n(n-1)(n-2)})\\ 
    N_{BBB} &= \frac{n}{n_B^3}(\frac{n_B}{n}-\frac{3n_B(n_B-1)}{n(n-1)}+\frac{2n_B(n_B-1)(n_B-2)}{n(n-1)(n-2)})\\ 
    N_{AAB} &= \frac{n}{n_A^2n_B}(-\frac{n_An_B}{n(n-1)}+\frac{n_A(n_A-1)n_B}{n(n-1)(n-2)})\\ 
    N_{Adj,A} & =  \frac{{n(n-1)(n-2)}}{(n_A-1)(n_A-2)n_A} \frac{n_A^3}{n^3}\\
    N_{Adj,B} & =  \frac{{v(n-1)(n-2)}}{(n_B-1)(n_B-2)n_B} \frac{n_B^3}{n^3}\\
\end{align*}
\subsection{Auxiliary Lemmas}
Let $x_i, y_i$ and $z_i$ be three possibly identical random variables such that $\bar{x}=\bar{y}=\bar{z}=0$. 
\begin{lemma}
\begin{equation*}
    E[\Bar{x}_A\Bar{y}_A\Bar{z}_A ] = N_{AAA}\frac{1}{n}\sum_{i=1}^n x_iy_iz_i
\end{equation*}
\begin{equation*}
    E[\Bar{x}_A\Bar{y}_A\Bar{z}_B ] =N_{AAB}\frac{1}{n}\sum_{i=1}^n x_iy_iz_i
\end{equation*}
\end{lemma}
\begin{proof}
We only prove the first equality. The second one can be proved analogously. First notice two useful equalities:
\begin{align*}
    & E[\sum_{i=1}^n T_ix_iy_i \sum_{j\not=i}T_jz_i] =E[\sum_{i=1}^n \sum_{j\not =i}T_iT_jx_iy_iz_j] = \sum_{i=1}^n\sum_{j\not =i} E[T_iT_j]x_iy_iz_j \\
    & = \frac{n_A(n_A-1)}{n(n-1)} \sum_{i=1}^n\sum_{j\not =i} x_iy_iz_j =  \frac{n_A(n_A-1)}{n(n-1)}( \sum_{i=1}^n\sum_{j=1}^n x_iy_iz_j -  \sum_{i=1}^Nx_iy_iz_i) = -\frac{n_A(n_A-1)}{n(n-1)}\sum_{i=1}^Nx_iy_iz_i,
\end{align*}
where in third equality uses the second moment estimate for a complete random experiments and the fifth equality uses the fact that $\sum_{i=1}^nz_i = n\bar{z}=0$.
\begin{align*}
    & E[\sum_{i=1}^nT_ix_i\sum_{j\not = i}T_jy_j \sum_{s\not \in \{i,j\}}T_sz_s] = \sum_{i=1}^n \sum_{j\not = i}\sum_{s\not \in \{i,j\}}E[T_iT_jT_s]x_iy_jz_s \\
    & = \frac{n_A(n_A-1)(n_A-2)}{n(n-1)(n-2)}\sum_{i=1}^n\sum_{j\not = i}\sum_{s\not \in \{i,j\}}x_iy_jz_s\\
    &= \frac{n_A(n_A-1)(n_A-2)}{n(n-1)(n-2)} (\sum_{i=1}^n \sum_{j\not = i}\sum_{s=1}^n x_iy_jz_s - \sum_{i=1}^n \sum_{j\not =i}x_iy_j(z_i+z_j))\\
    &= \frac{n_A(n_A-1)(n_A-2)}{n(n-1)(n-2)} (- \sum_{i=1}^n \sum_{j=1}^n x_iy_j(z_i+z_j)+ 2\sum_{i=1}^n x_iy_iz_i)\\
    &= \frac{2n_A(n_A-1)(n_A-2)}{n(n-1)(n-2)} \sum_{i=1}^n x_iy_iz_i,
\end{align*}
where the fourth and fifth equality use $\sum_{i=1}^nx_i=\sum_{i=1}^ny_i=\sum_{i=1}^nz_i=0$.
Finally,
\begin{align*}
   & E[\Bar{x}_A\Bar{y}_A\Bar{z}_A ] = \frac{1}{n_A^3} E[\sum_{i=1}^n T_ix_i\sum_{i=1}^n T_iy_i\sum_{i=1}^n T_iz_i] \\
   & = \frac{1}{n_A^3} (E[\sum_{i}T_ix_iy_iz_i] + E[\sum_{i=1}^n\sum_{j\not =i}T_iT_j(x_iy_iz_j + x_iy_jz_i+ x_jy_iz_i)] + E[\sum_{i=1}^NT_ix_i\sum_{j\not = i}T_jy_j \sum_{s\not \in \{i,j\}}T_sz_s] \\
   & = \frac{1}{n_A^3} (\frac{n_A}{n}-\frac{3n_A(n_A-1)}{n(n-1)} + \frac{2n_A(n_A-1)(n_A-2)}{n(n-1)(n-2)} )\sum_{i=1}^n x_iy_iz_i,
\end{align*}
where for the last equality we apply the previous two equalities.
\end{proof}
\begin{lemma}
\begin{equation}
 N_{Adj,A}E[\frac{1}{n_A}\sum_{i\in[A]}(x_i-\bar{x}_A)(y_i-\bar{y}_A)(z_i-\bar{z}_A)] =  \frac{1}{n}\sum_{i=1}^n (x_i-\bar{x})(y_i-\bar{y})(z_i-\bar{z})
\end{equation}
\end{lemma}
\begin{proof}
\begin{align*}
    & n^3\sum_{i=1}^n (x_i-\Bar{x})(y_i-\Bar{y})(z_i-\Bar{z})= \sum_{i=1}^n\sum_{j=1}^n\sum_{k=1}^n\sum_{s=1}^n (x_i-x_j)(y_i-y_k)(z_i-z_s)  \\
    & = n^3 \sum_{i=1}^n x_iy_iz_i - n^2 \sum_{i=1}^n\sum_{j=1}^n (x_iy_ix_j+x_iy_jz_j+x_jy_iz_i) +2n\sum_i\sum_j\sum_s x_iy_jz_s \\
    & = (n^3-3n^2+2n)\sum_{i=1}^n x_iy_iz_i - (n^2-2n)\sum_{i=1\not=j} (x_iy_iz_j+x_iy_jz_j+x_jy_iz_i) +2n\sum_{i\not= j\not=s} x_iy_jz_s 
\end{align*}
Consider $n_A^4$ times the third moment estimator $\frac{1}{n_A}\sum_{i\in[A]}(x_i-\bar{x}_A)(y_i-\bar{y}_A)(z_i-\bar{z}_A)$:
\begin{align*}
    &= (n_A^3-3n_A^2+2n_A)\sum_{i=1}^n T_ix_iy_iz_i - (n_A^2-2n_A)\sum_{i=1\not=j}T_iT_j (x_iy_iz_j+x_iy_jz_j+x_jy_iz_i) +2n_A\sum_{i\not= j\not=s}T_iT_jT_s x_iy_jz_s \\
    &= n_A(n_A-1)(n_A-2)\sum_{i=1}^n T_ix_iy_iz_i - n_A(n_A-2)\sum_{i=1\not=j}T_iT_j (x_iy_iz_j+x_iy_jz_j+x_jy_iz_i) +2n_A\sum_{i\not= j\not=s}T_iT_jT_s x_iy_jz_s \\
\end{align*}
In expectation this equals to:
\begin{align*}
   &= \frac{n_A(n_A-1)(n_A-2)n_A}{n}\sum_{i=1}^n x_iy_iz_i - \frac{n_A(n_A-1)(n_A-2)n_A}{n(n-1)}\sum_{i=1\not=j} (x_iy_iz_j+x_iy_jz_j+x_jy_iz_i)\\ 
   & +2\frac{n_A(n_A-1)(n_A-2)n_A}{n(n-1)(n-2)}\sum_{i\not= j\not=s} x_iy_jz_s \\
   &= \frac{n_A(n_A-1)(n_A-2)n_A}{n(n-1)(n-2)n} ((n^3-3n^2+2n)\sum_{i=1}^n x_iy_iz_i - (n^2-2n)\sum_{i=1\not=j} (x_iy_iz_j+x_iy_jz_j+x_jy_iz_i)\\ 
   & +2N\sum_{i\not= j\not=s} x_iy_jz_s) \\
   & = \frac{n_A(n_A-1)(n_A-2)n_A}{n(n-1)(n-2)n}\times n^3\sum_{i=1}^n (x_i-\Bar{x})(y_i-\Bar{y})(z_i-\Bar{z})
\end{align*}
\end{proof}
\subsection{Theorem 3.1}
\begin{proof}
By the Frisch–Waugh–Lovell theorem, the OLS estimate of the coefficient can be written $\Q=\D^{-1}\N$, where
\begin{align*}
\N&=\pA (\azA-\aA\zA) +\pB (\bzB-\bB \zB)\\
&=\pA (\aszA-\asA\zA) +\pB (\bszB-\bsB \zB)\\
&= N+\pA (\aszA-\asz) +\pB (\bszB-\bsz)  -\pA \asA\zA -\pB \bsB\zB
\end{align*}
and
\begin{align*}
\D= & \zz - \pA \zA  \zA' - \pB \zB  \zB' .
\end{align*}
\begin{align*}
	\Q = & D^{-1} \N +\left(\D^{-1}-D^{-1}\right) \N
	\\ = & D^{-1} \left(N+\pA (\aszA-\asz) +\pB (\bszB-\bsz) - \pA \asA\zA -\pB \bsB\zB \right) +\left(\D^{-1}-D^{-1}\right) \N
	\\ = & Q + D^{-1}\left(\pA (\aszA-\asz) +\pB (\bszB-\bsz)\right)  - D^{-1}\left(\pA \asA\zA +\pB \bsB\zB \right)
	\\ &  +\left(\D^{-1}-D^{-1}\right) \N
	\\ =&  Q+ \nu_1 +\nu_3 +\nu_2,
\end{align*}
Now, consider the order of the terms $\nu_1$, $\nu_2$ and $\nu_3$.  For $\nu_1$, $D^{-1}=\zz$, $\pA$ and $\pB$ are $O(1)$ by assumption.  Meanwhile, $(\aszA-\az)$ and $(\bsz_B-\bz)$ are each $O_p(n^{-\frac{1}{2}})$ by moment conditions on $a$, $b$ and $z$.  Therefore $\nu_1 = O_p(n^{-\frac{1}{2}})$\\
For $\nu_2$, note that $\D - D=- \pA \zA  \zA' - \pB \zB  \zB' $ which is $O_p(n^{ -1})$.  $\D^{-1} - D^{-1}=\D^{-1} (D-\D) D^{-1} = O_p(1)O_p(n^{-1})O(1)$ is also $O_p(n^{-1})$.  Meanwhile, $\N$ converges to a constant vector so that it is $O_p(1)$. Therefore $\nu_2 = O_p(n^{-1})$\\
For $\nu_3$, $\zA$, $\zB$, $\bsB$ and $\asA$ are $O_p(n^{-\frac{1}{2}})$, and $D^{-1}$, $\pA$ and $\pB$ are $O_p(1)$. Therefore $\nu_3 = O_p(n^{-1})$
\end{proof}
\subsection{Corollary 3.1}
\begin{proof}
		\begin{align*}
	\E[ \ateAdj -ATE ] =&  \E \left[ \aA- \bB -(\zA-\zB)'\Q \right] -ATE
	\\ = &  \left(\E \left[ \aA- \bB\right] -ATE \right)- \E \left[ (\zA-\zB)'\right]Q  - \E \left[ (\zA-\zB)'(\nu_1+\nu_2+\nu_3) \right]
		\\ = &  \left(ATE  -ATE \right)- 0 \times Q  - \E \left[ (\zA-\zB)'(\nu_1+\nu_2+\nu_3) \right]  
	\\ = & -\E \left[ (\zA-\zB)'(\nu_1+\nu_2+\nu_3) \right].
	\end{align*}
where in the third equality we used the unbiasedness of the difference-in-means estimator and the unbiasedness of $\zA$ and $\zB$ as estimators of the sample mean $\zBar$. 
\end{proof}
\subsection{Theorem 3.2}
\begin{proof}
        The decomposition is algebraic. Now $\hat{N}_A-N_A=O_p(n^{-\frac{1}{2}})$,  $\hat{N}_B-N_B=O_p(n^{-\frac{1}{2}})$, $\hat{N}_A=O_p(1)$ and $\hat{N}_B=O_p(1)$ by the assumptions. \\
        $$\hat{D}^{-1}_A - D^{-1}  = \hat{D}_A^{-1} (D-\hat{D}_A) D^{-1} = O_p(1)O_p(n^{-\frac{1}{2}})O(1)$$
        Similarly $\hat{D}_B^{-1}-D^{-1}=O_p(n^{-\frac{1}{2}})$. These estimates give the orders in the theorem.
\end{proof}
\subsection{Corollary 3.2}
\begin{proof}
        Same as Corollary 3.1 after noting that $\zA=O_p(n^{-\frac{1}{2}})$ and $\zB=O_p(n^{-\frac{1}{2}})$
\end{proof}
\subsection{Theorem 4.1}
Let
\begin{equation*}
C_{A,NI} = \frac{n_A}{n_B} \frac{n}{n_A^3}(\frac{n_A}{n}-\frac{3n_A(n_A-1)}{n(n-1)}+\frac{2n_A(n_A-1)(n_A-2)}{n(n-1)(n-2)})\frac{{n(n-1)(n-2)}}{(n_A-1)(n_A-2)n_A} \frac{n_A^3}{n^3} 
\end{equation*}
\begin{equation*}
    C_{B,NI}= \frac{n_A}{n_B}\frac{n}{n_A^2n_B}(-\frac{n_An_B}{n(n-1)}+\frac{n_A(n_A-1)n_B}{n(n-1)(n-2)}) \frac{{n(n-1)(n-2)}}{(n_A-1)(n_A-2)n_A} \frac{n_A^3}{n^3}  
\end{equation*}
Note both of them are of order $O(n^{-1})$.
\begin{proof}
We first propose an estimator for $	\E[(\zB-\zA)\nu_1]$. Using the assumption $\zbar=0$, we have 
	\begin{align*}
	(\zB-\zA)= \frac{1}{\pA}\zB=-\frac{1}{\pB}\zA .
	\end{align*}
	Then we have
	\begin{align*}
	\E[(\zB-\zA)'\nu_1]& = \frac{1}{\pA}\E\left[ \zB' D^{-1} \left( \pA (\aszA-\asz) + \pB (\bszB-\bsz) \right) \right] \\
	& = \frac{1}{n-1}(\frac{1}{n}\sum_{i=1}^n z_i'D^{-1}b_i^*z_i -\frac{1}{n}\sum_{i=1}^n z_iD^{-1}a_i^*z_i)\\
	& = \frac{1}{n-1}(\overline{hb} - \bar{h}\bar{b} ) -\frac{1}{n-1}(\overline{ha} - \bar{h}\bar{a} ) 
 	\end{align*}
    where the second equality follows from Proposition 1, \cite{freedman2008B}. A estimator of this bias is:
    \begin{equation*}
        \frac{1}{n-1} \Large[ \frac{n_B \left( n-1\right) }{\left(n_B-1\right) n}\left( \overline{hb}_B- \overline{h}_B\overline{b}_B \right)-\frac{n_A \left(n-1\right) }{\left(n_A-1\right) n}\left( \overline{ha}_A -\overline{h}_A \overline{a}_A \right) \Large]
    \end{equation*}
    It is clear that this adjustment is of order $O_p(n^{-1})$.\\
    It should be obvious that $(\zA-\zB)'\nu_2$ is directly estimable and of order $O_p(n^{-1})$. \\ 
    We now propose an estimator for $\E[(\zB-\zA)'\nu_3]$:
    \begin{align*}
        \E[(\zB-\zA)'\nu_3] & = \frac{1}{p_B} E[\zA' D^{-1}p_A \asA \zA] + \frac{1}{p_B} E[\zA' D^{-1}p_B \bsB \zB]\\
        & = \frac{1}{p_B} E[\zA' D^{-1}p_A \asA \zA] - \frac{1}{p_B} E[\zA' D^{-1}p_A \bsB \zA]\\
        & =  \frac{p_A}{p_B} N_{AAA} \frac{1}{n}\sum_{i=1}^n z_i'D^{-1}z_i a^*_i -  \frac{p_A}{p_B} N_{AAB}\frac{1}{n} \sum_{i=1}^n z_i'D^{-1}z_i b^*_i 
    \end{align*}
By Lemma A1 and A2, an unbiased estimator for this quantity is:
\begin{align*}
        &\frac{n_A}{n_B} N_{AAA} N_{adj,A} \frac{1}{n_A}\sum_{i\in[A]} z_i'D^{-1}z_i (a_i-\aA) -  \frac{n_A}{n_B} N_{AAB} N_{adj,B}  \sum_{i\in[B]} z_i'D^{-1}z_i (b_i-\bB)\\
        & = \frac{C_{A,NI}}{n_A}\sum_{i\in[A]} z_i'D^{-1}z_i (a_i-\aA) - \frac{C_{B,NI}}{n_B} \sum_{i\in[B]} z_i'D^{-1}z_i (b_i-\bB)
\end{align*}
    It is clear that this adjustment is of order $O_p(n^{-1})$.\\
\end{proof}

\subsection{Theorem 4.2}
We have:
\begin{equation*}
    C_A = \frac{n}{n_A^3}(\frac{n_A}{n}-\frac{3n_A(n_A-1)}{n(n-1)}+\frac{2n_A(n_A-1)(n_A-2)}{n(n-1)(n-2)})\frac{{n(n-1)(n-2)}}{(n_A-1)(n_A-2)n_A} \frac{n_A^3}{n^3}  
\end{equation*}
\begin{equation*}
    C_B = \frac{n}{n_B^3}(\frac{n_B}{n}-\frac{3n_B(n_B-1)}{n(n-1)}+\frac{2n_B(n_B-1)(n_B-2)}{n(n-1)(n-2)})\frac{{n(n-1)(n-2)}}{(n_B-1)(n_B-2)n_B} \frac{n_B^3}{n^3}  
\end{equation*}
Note that $C_A$ and $C_B$ are of order $O(n^{-1})$
\begin{proof}
We prove the result for bias in the control arm. Proof for the treated arm is analogous. First notice
\begin{align*}
    E[\bar{z}_B'(\nu_{1B}+\nu_{2B})] & = E[\bar{z}_B'(\hat{D}^{-1}_B - D^{-1}) \hat{N}_B] \\
    & + E[\bar{z}_B'D^{-1} \frac{1}{n_B}\large( \sum_{i\in [B]}z_ib_i^* -\frac{1}{n}\sum_{i\in [N]}z_ib_i^*\large)]\\
    & - E[\bar{z}_B'D^{-1} \bar{z}_B \bar{b}^*_B]
\end{align*}
The first term is directly estimable and is of order $O_p(n^{-1})$. Analogous to the noninteractive case, the second term can be estimated by:
\begin{equation*}
    \frac{1}{n-1} \frac{n_A}{n_B} \frac{n_B \left( n-1\right) }{\left(n_B-1\right) n}\left( \overline{hb}_B- \overline{h}_B\overline{b}_B \right) = \frac{1}{n}\frac{n_A}{n_B-1} \left( \overline{hb}_B- \overline{h}_B\overline{b}_B \right)
\end{equation*}
It should be clear that this term is of order $O_p(n^{-1})$. \\
The third term can be estimated by:
\begin{align*}
    N_{BBB}*N_{adj,B}\frac{1}{n_B} \sum_{i\in[B]} z_i'D^{-1}z_i (b_i-\bB) = C_{B,I} \sum_{i\in[B]} z_i'D^{-1}z_i (b_i-\bB) 
\end{align*}
It is clear this term is of order $O_p(n^{-1})$
\end{proof}
\section{More Simulation Results}
\subsection{Details on Simulation Schemes}
\begin{table}[H] 
\centering 
\begin{tabular}{l c c c c c c } 
\toprule 
\cmidrule(l){2-6} 
 & $X_1(i)$ & $X_2(i)$ & $Y_0(i)$ & $Y_1(i)$ & \#treated & ATE \\ 
\midrule 
\textbf{Scheme 1, N=24}& & & & & &  \\
DGP1.1 &\multirow{3}{*}{Beta(0.5,0.5)}& \multirow{3}{*}{Tri(0,1)} & 0 & 2$h_i$ & \multirow{3}{*}{8} &  \multirow{3}{*}{0}\\
DGP1.2 & & & -$h_i$ & $h_i$ &&\\
DGP1.3& & & $h_i$ & $h_i$ && \\
& & & & \\
\textbf{Scheme 2, N=24}& & & &  \\
DGP2.1 &\multirow{3}{*}{Beta(2,5)}& \multirow{3}{*}{Norm(0,1)} & 0 & 2$h_i$ & \multirow{3}{*}{8} &  \multirow{3}{*}{0} \\
DGP2.2& & & -$h_i$ & $h_i$&&\\
DGP2.3 & & & $h_i$ & $h_i$&&\\
& & & & \\
\textbf{Scheme 3, N=24}& & & &  \\
DGP3.1 &\multirow{3}{*}{Uniform(0,1)}& \multirow{3}{*}{Uniform(0,1), Squared, Reversed} & 0 & 2$h_i$ & \multirow{3}{*}{8} &  \multirow{3}{*}{0} \\
DGP3.2& & & -$h_i$ & $h_i$&&\\
DGP3.3 & & & $h_i$ & $h_i$&&\\
& & & & \\
\textbf{Scheme 4, N=24}& & & &  \\
DGP3.1 &\multirow{3}{*}{Uniform(0,1)}& \multirow{3}{*}{Uniform(0,1), Squared} & 0 & 2$h_i$ & \multirow{3}{*}{8} &  \multirow{3}{*}{0} \\
DGP3.2& & & -$h_i$ & $h_i$&&\\
DGP3.3 & & & $h_i$ & $h_i$&&\\
& & & & \\
\bottomrule
\end{tabular}
\caption{DGPs for Simulations. Beta($\alpha$,$\beta$) are the beta distributions with shape parameters $\alpha$ and $\beta$. Tri(0,1) is the symmetric triangular distribution on the unit interval. Norm(0,1) is the standard normal distribution. Uniform(0,1) is the uniform distribution on the unit interval. For Uniform(0,1) Squared, we square the quantiles of the uniformt distribution. For Uniform(0,1), Squared, Reversed, we reverse the order of the quantiles such that the 1st unit corresponds to the highest quantiles. $h_i$ is the studentized leverage ratio for the $i$th unit. It is computed as $h_i=\frac{v_i-\bar{v}}{\sigma_v}$, where $v_i=x_i'(\sum_{i=1}^nx_ix_i')^{-1}x_i$, $\bar{v}=\frac{1}{n}\sum_{i=1}^nv_i$ and $\sigma^2_v=\frac{1}{n-1}\sum_{i=1}^N(v_i-\bar{v}^2)$. }
\end{table}
\newpage
\begin{landscape}
\begin{figure}
    \centering
\includegraphics[scale=0.1]{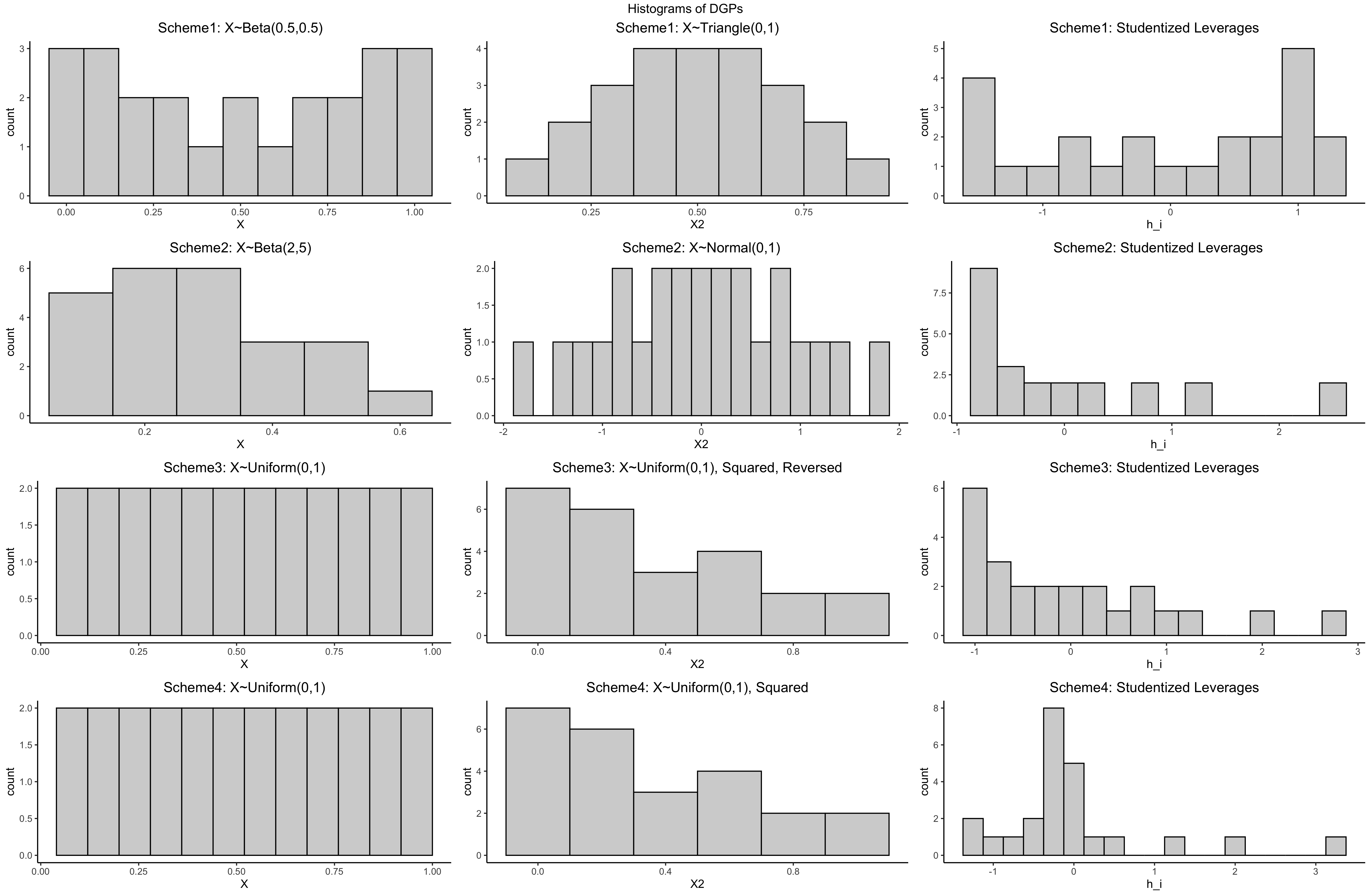}\\
    \caption{Histogram Plots for the DGPs. For each scheme, we plot the distributions of $X$, $X2$ and the studentized leverages $h_i$. Please refer to the note under Table 1 or Table 5 for the definition of $h_i$s.}
    \label{fig:my_label}
\end{figure}
\end{landscape}
\newpage
\subsection{Scheme 1}

\begin{table}[H]
\centering
\begin{threeparttable}
\caption{Confidence Interval Statistics for Scheme 1 using HC2}
\begin{tabular}{lcccccc}
\toprule 
& \multicolumn{5}{c}{\textbf{ATE Estimators}} \\ 
\cmidrule(l){2-7} 
  & \multicolumn{3}{c}{HC2} & \multicolumn{3}{c}{BC-HC2} \\ 
  &  Z & Student-t & Satterthwaite & Z & Student-t & Satterthwaite  \\ 
  \hline
  \hline
&&&&&\\
 \textbf{DGP1.1, N =24}& & & & & \\
\textit{Coverage, Percentage}& & & & & \\
Non-interacted (Debiased) & 95.30 & 96.00 & 96.60 & 95.40 & 96.10 & 96.70 \\ 
  Interacted (Debiased) & 94.60 & 95.30 & 97.00 & 95.00 & 95.70 & 97.30 \\ 
&&&&&&\\
\textit{CI Width, Average}& & & & & \\
Non-interacted (Debiased) & 2.68 & 2.83 & 3.00 & 2.70 & 2.85 & 3.02 \\ 
  Interacted (Debiased)& 3.20 & 3.38 & 4.13 & 3.26 & 3.44 & 4.28 \\ 
 &&&&&&\\
\textit{CI Width, Median}& & & & & \\
Non-interacted (Debiased) & 2.71 & 2.86 & 3.01 & 2.71 & 2.86 & 3.01 \\ 
   Interacted (Debiased)& 3.02 & 3.19 & 3.67 & 3.02 & 3.19 & 3.67 \\ 
&&&&&\\
 \textbf{DGP1.2, N =24}& & & & & \\
 \textit{Coverage, Percentage}& & & & & \\
Non-interacted (Debiased) & 99.90 & 100.00 & 100.00 & 99.90 & 100.00 & 100.00 \\ 
  Interacted (Debiased) & 99.90 & 99.90 & 100.00 & 100.00 & 100.00 & 100.00 \\ 
&&&&&&\\
\textit{CI Width, Average}& & & & & \\
Non-interacted (Debiased) &1.84 & 1.95 & 2.06 & 1.85 & 1.95 & 2.06 \\ 
   Interacted (Debiased)& 1.87 & 1.98 & 2.42 & 1.90 & 2.00 & 2.48 \\ 
&&&&&&\\
\textit{CI Width, Median}& & & & & \\
 Non-interacted (Debiased) & 1.83 & 1.93 & 2.03 & 1.83 & 1.93 & 2.03 \\ 
  Interacted (Debiased) & 1.77 & 1.87 & 2.16 & 1.77 & 1.87 & 2.16 \\ 
&&&&&&\\
 \textbf{DGP1.3, N =24}& & & & & \\
 \textit{Coverage, Percentage}& & & & & \\
Non-interacted (Debiased) & 93.50 & 94.60 & 95.60 & 93.60 & 94.70 & 95.60 \\ 
   Interacted (Debiased)& 93.80 & 94.80 & 97.00 & 94.10 & 95.00 & 97.10 \\ 
&&&&&&\\
\textit{CI Width, Average}& & & & & \\
Non-interacted (Debiased) & 1.63 & 1.72 & 1.82 & 1.63 & 1.72 & 1.82 \\ 
    Interacted (Debiased) & 1.87 & 1.98 & 2.42 & 1.91 & 2.01 & 2.49 \\
  &&&&&&\\
\textit{CI Width, Median}& & & & & \\
Non-interacted (Debiased) & 1.62 & 1.71 & 1.80 & 1.62 & 1.71 & 1.80 \\ 
   Interacted (Debiased) & 1.77 & 1.87 & 2.16 & 1.77 & 1.87 & 2.16 \\
&&&&&&\\
   \hline
\end{tabular}
 \begin{tablenotes}
      \small
\item Note: CI Coverage measures the empirical coverage rates of nominal 95 percent confidence intervals.  Both Student-t and Satterthwaite adjustments are calculated using the R-package \texttt{clubSandwich}\cite{clubSandwich}.
\end{tablenotes}
\end{threeparttable}

\end{table}

\begin{table}[H]
\centering
\begin{threeparttable}
\caption{Confidence Interval Statistics for Scheme 1 using HC3}
\begin{tabular}{lcccccc}
\toprule 
& \multicolumn{5}{c}{\textbf{ATE Estimators}} \\ 
\cmidrule(l){2-7} 
  & \multicolumn{3}{c}{HC3} & \multicolumn{3}{c}{BC-HC3} \\ 
  &  Z & Student-t & Satterthwaite & Z & Student-t & Satterthwaite  \\ 
  \hline
  \hline
&&&&&\\
 \textbf{DGP1.1, N =24}& & & & & \\
\textit{Coverage, Percentage}& & & & & \\
Non-interacted (Debiased) & 96.70 & 97.20 & 97.70 & 96.70 & 97.20 & 97.80 \\ 
    Interacted (Debiased)& 97.40 & 97.70 & 98.80 & 97.60 & 97.90 & 99.00 \\ 
 &&&&&\\
\textit{CI Width, Average}& & & & & \\
Non-interacted (Debiased) & 3.00 & 3.16 & 3.37 & 3.01 & 3.18 & 3.39 \\ 
     Interacted (Debiased) & 4.58 & 4.83 & 7.60 & 4.60 & 4.85 & 7.83 \\ 
&&&&&\\
 \textit{CI Width, Median}& & & & & \\
Non-interacted (Debiased) & 3.02 & 3.18 & 3.38 & 3.03 & 3.20 & 3.40 \\ 
  Interacted (Debiased) & 3.94 & 4.16 & 5.18 & 3.90 & 4.11 & 5.12 \\
&&&&&\\
 \textbf{DGP1.2, N =24}& & & & & \\ 
\textit{Coverage, Percentage}& & & & & \\
Non-interacted (Debiased) &100.00 & 100.00 & 100.00 & 100.00 & 100.00 & 100.00 \\ 
  Interacted (Debiased) & 100.00 & 100.00 & 100.00 & 100.00 & 100.00 & 100.00 \\
&&&&&\\
\textit{CI Width, Average}& & & & & \\
Non-interacted (Debiased) & 2.06 & 2.17 & 2.32 & 2.06 & 2.17 & 2.32 \\ 
  Interacted (Debiased) & 2.54 & 2.68 & 4.16 & 2.54 & 2.68 & 4.25 \\ 
 &&&&&\\
 \textit{CI Width, Median}& & & & & \\
Non-interacted (Debiased) & 2.04 & 2.15 & 2.28 & 2.04 & 2.15 & 2.28 \\ 
 I.Nobias & 2.21 & 2.33 & 2.91 & 2.19 & 2.31 & 2.88 \\
&&&&&\\
 \textbf{DGP1.3, N =24}& & & & & \\ 
\textit{Coverage, Percentage}& & & & & \\
Non-interacted (Debiased) &95.30 & 96.10 & 97.00 & 95.30 & 96.20 & 97.00 \\ 
  Interacted (Debiased) & 96.80 & 97.40 & 98.90 & 97.00 & 97.50 & 99.00 \\ 
&&&&&\\
\textit{CI Width, Average}& & & & & \\
Non-interacted (Debiased) & 1.78 & 1.88 & 2.01 & 1.79 & 1.89 & 2.01 \\ 
  Interacted (Debiased)& 2.54 & 2.68 & 4.16 & 2.57 & 2.71 & 4.33 \\ 
&&&&&\\
 \textit{CI Width, Median}& & & & & \\
Non-interacted (Debiased) & 1.77 & 1.87 & 1.98 & 1.77 & 1.87 & 1.99 \\ 
  Interacted (Debiased)& 2.21 & 2.33 & 2.91 & 2.20 & 2.32 & 2.89 \\ 
   \hline
\end{tabular}
 \begin{tablenotes}
      \small
\item Note: CI Coverage measures the empirical coverage rates of nominal 95 percent confidence intervals. HC3 refers to the approximate jackknife robust standard error estimator\cite{mackinnon1985some}. Both Student-t and Satterthwaite adjustments are calculated using the R-package \texttt{clubSandwich}\cite{clubSandwich}.
\end{tablenotes}
\end{threeparttable}

\end{table}

\subsection{Scheme 2}

\begin{table}[H]
\centering
\begin{threeparttable}
\caption{Confidence Interval Statistics for Scheme 2 using HC2}
\begin{tabular}{lcccccc}
\toprule 
& \multicolumn{5}{c}{\textbf{ATE Estimators}} \\ 
\cmidrule(l){2-7} 
  & \multicolumn{3}{c}{HC2} & \multicolumn{3}{c}{BC-HC2} \\ 
  &  Z & Student-t & Satterthwaite & Z & Student-t & Satterthwaite  \\ 
  \hline
  \hline
&&&&&\\
 \textbf{DGP2.1, N =24}& & & & & \\
\textit{Coverage, Percentage}& & & & & \\
Non-interacted (Debiased) & 91.40 & 92.30 & 92.80 & 91.40 & 92.30 & 92.80 \\ 
   Interacted (Debiased) & 44.70 & 47.00 & 54.80 & 85.40 & 87.60 & 93.20 \\ 
&&&&&&\\
\textit{CI Width, Average}& & & & & \\
  Non-interacted (Debiased) & 1.93 & 2.04 & 2.16 & 2.04 & 2.15 & 2.28 \\ 
  Interacted (Debiased)& 0.60 & 0.63 & 0.89 & 2.48 & 2.61 & 4.50 \\ 
 &&&&&&\\
\textit{CI Width, Median}& & & & & \\
Non-interacted (Debiased) & 1.95 & 2.06 & 2.17 & 1.95 & 2.06 & 2.17 \\ 
  Interacted (Debiased) & 0.51 & 0.54 & 0.67 & 0.51 & 0.54 & 0.67 \\ 
&&&&&\\
 \textbf{DGP2.2, N =24}& & & & & \\
 \textit{Coverage, Percentage}& & & & & \\
Non-interacted (Debiased) &96.30 & 96.70 & 96.90 & 96.10 & 96.50 & 96.80 \\ 
  Interacted (Debiased) & 58.90 & 61.40 & 72.40 & 95.30 & 96.70 & 99.60 \\ 
&&&&&&\\
\textit{CI Width, Average}& & & & & \\
Non-interacted (Debiased) &1.91 & 2.02 & 2.14 & 1.96 & 2.07 & 2.19 \\ 
  Interacted (Debiased) & 0.40 & 0.42 & 0.59 & 1.33 & 1.40 & 2.36 \\ 
&&&&&&\\
\textit{CI Width, Median}& & & & & \\
Non-interacted (Debiased) &1.94 & 2.05 & 2.16 & 1.94 & 2.05 & 2.16 \\ 
  Interacted (Debiased) & 0.38 & 0.40 & 0.48 & 0.38 & 0.40 & 0.48 \\ 
&&&&&&\\
 \textbf{DGP2.3, N =24}& & & & & \\
 \textit{Coverage, Percentage}& & & & & \\
Non-interacted (Debiased) & 62.50 & 65.40 & 68.30 & 76.50 & 80.90 & 85.00 \\ 
  Interacted (Debiased) & 54.20 & 57.00 & 67.80 & 87.40 & 89.60 & 94.40 \\ 
&&&&&&\\
\textit{CI Width, Average}& & & & & \\
Non-interacted (Debiased) & 0.37 & 0.39 & 0.41 & 0.45 & 0.48 & 0.51 \\ 
  Interacted (Debiased) & 0.40 & 0.42 & 0.59 & 1.33 & 1.41 & 2.38 \\ 
  &&&&&&\\
\textit{CI Width, Median}& & & & & \\
Non-interacted (Debiased) & 0.35 & 0.37 & 0.40 & 0.35 & 0.37 & 0.40 \\ 
  Interacted (Debiased) & 0.38 & 0.40 & 0.48 & 0.38 & 0.40 & 0.48 \\ 
&&&&&&\\
   \hline
\end{tabular}
 \begin{tablenotes}
      \small
\item Note: CI Coverage measures the empirical coverage rates of nominal 95 percent confidence intervals.  Both Student-t and Satterthwaite adjustments are calculated using the R-package \texttt{clubSandwich}\cite{clubSandwich}.
\end{tablenotes}
\end{threeparttable}

\end{table}

\begin{table}[H]
\centering
\begin{threeparttable}
\caption{Confidence Interval Statistics for Scheme 2 using HC3}
\begin{tabular}{lcccccc}
\toprule 
& \multicolumn{5}{c}{\textbf{ATE Estimators}} \\ 
\cmidrule(l){2-7} 
  & \multicolumn{3}{c}{HC3} & \multicolumn{3}{c}{BC-HC3} \\ 
  &  Z & Student-t & Satterthwaite & Z & Student-t & Satterthwaite  \\ 
  \hline
  \hline
&&&&&\\
 \textbf{DGP2.1, N =24}& & & & & \\
\textit{Coverage, Percentage}& & & & & \\
Non-interacted (Debiased) & 93.20 & 94.10 & 94.50 & 93.20 & 94.00 & 94.50 \\ 
  Interacted (Debiased) & 74.20 & 76.00 & 87.60 & 96.40 & 97.00 & 98.80 \\ 
 &&&&&\\
\textit{CI Width, Average}& & & & & \\
Non-interacted (Debiased)  & 2.37 & 2.50 & 2.67 & 2.44 & 2.57 & 2.76 \\ 
  Interacted (Debiased)  & 1.60 & 1.69 & 5.29 & 12.18 & 12.85 & 57.49 \\ 
&&&&&\\
 \textit{CI Width, Median}& & & & & \\
Non-interacted (Debiased)  & 2.37 & 2.50 & 2.66 & 2.44 & 2.57 & 2.73 \\ 
  Interacted (Debiased) & 1.07 & 1.13 & 1.85 & 2.11 & 2.23 & 3.94 \\ 
&&&&&\\
 \textbf{DGP2.2, N =24}& & & & & \\ 
\textit{Coverage, Percentage}& & & & & \\
Non-interacted (Debiased)  & 97.00 & 97.30 & 97.40 & 96.80 & 97.10 & 97.30 \\ 
  Interacted (Debiased)  & 85.30 & 86.70 & 96.10 & 99.90 & 99.90 & 100.00 \\ 
&&&&&\\
\textit{CI Width, Average}& & & & & \\
Non-interacted (Debiased)  & 2.34 & 2.47 & 2.64 & 2.35 & 2.48 & 2.66 \\ 
  Interacted (Debiased)  & 0.93 & 0.98 & 2.88 & 6.26 & 6.60 & 29.16 \\ 
 &&&&&\\
 \textit{CI Width, Median}& & & & & \\
Non-interacted (Debiased)  & 2.33 & 2.46 & 2.62 & 2.37 & 2.50 & 2.65 \\ 
  Interacted (Debiased)  & 0.67 & 0.71 & 1.16 & 1.18 & 1.25 & 2.21 \\ 
&&&&&\\
 \textbf{DGP2.3, N =24}& & & & & \\ 
\textit{Coverage, Percentage}& & & & & \\
Non-interacted (Debiased) & 71.60 & 74.60 & 78.10 & 86.90 & 90.40 & 93.30 \\ 
  Interacted (Debiased)  & 83.30 & 84.70 & 94.50 & 98.20 & 98.80 & 99.60 \\ 
&&&&&\\
\textit{CI Width, Average}& & & & & \\
Non-interacted (Debiased) & 0.45 & 0.47 & 0.50 & 0.54 & 0.57 & 0.61 \\ 
  Interacted (Debiased)  & 0.93 & 0.98 & 2.88 & 6.28 & 6.63 & 29.31 \\ 
&&&&&\\
 \textit{CI Width, Median}& & & & & \\
Non-interacted (Debiased) & 0.42 & 0.45 & 0.48 & 0.50 & 0.52 & 0.55 \\ 
  Interacted (Debiased)  & 0.67 & 0.71 & 1.16 & 1.18 & 1.25 & 2.24 \\  
   \hline
\end{tabular}
 \begin{tablenotes}
      \small
\item Note: CI Coverage measures the empirical coverage rates of nominal 95 percent confidence intervals. HC3 refers to the approximate jackknife robust standard error estimator\cite{mackinnon1985some}. Both Student-t and Satterthwaite adjustments are calculated using the R-package \texttt{clubSandwich}\cite{clubSandwich}.
\end{tablenotes}
\end{threeparttable}

\end{table}

\subsection{Scheme 3}
\begin{table}[H]
\centering
\caption{Simulation Results for Scheme 3}    
\begin{threeparttable}
\begin{tabular}{lccccc}
\toprule 
& \multicolumn{5}{c}{\textbf{ATE Estimators}} \\ 
\cmidrule(l){2-6} 
 & Unadjusted & \multicolumn{2}{c}{OLS} & \multicolumn{2}{c}{Debiased} \\ 
  & & Non-Int. & Interacted & Non-Int. & Interacted \\ 
\midrule 
\textbf{DGP3.1, N =24}& & & & & \\
Bias & 0.000 & -0.144 & -0.004 & 0.000 & 0.000 \\
  SD & 0.577 & 0.362 & 0.281 & 0.433 & 0.387 \\ 
  MSE & 0.577 & 0.390 & 0.281 & 0.433 & 0.387 \\ 
 \makecell[l]{CI Coverage  (HC2, Student-t)}& 0.943 & 0.949 & 0.843 & 0.959 & 0.712 \\
   \makecell[l]{CI Coverage  (HC2, Satterthwaite)} & 0.948 & 0.959 & 0.917 & 0.965 & 0.810 \\
  \makecell[l]{CI Coverage  (BC-HC2, Student-t)} &  &  &  & 0.960 & 0.876 \\ 
  \makecell[l]{CI Coverage  (BC-HC2, Satterthwaite)} &  &  &  & 0.966 & 0.936 \\ 
&&&&&\\
\textbf{DGP3.2, N =24}& & & & & \\
Bias & 0.000 & 0.144 & 0.003 & 0.000 & 0.000 \\ 
  SD & 0.144 & 0.342 & 0.129 & 0.316 & 0.198 \\ 
  MSE & 0.144 & 0.371 & 0.129 & 0.316 & 0.198 \\ 
 \makecell[l]{CI Coverage  (HC2, Student-t)} & 1.000 & 0.963 & 0.927 & 0.983 & 0.804 \\
  \makecell[l]{CI Coverage  (HC2, Satterthwaite)}  &1.000 &0.967 & 0.980  & 0.985 & 0.906 \\
  \makecell[l]{CI Coverage  (BC-HC2, Student-t)} &  &  & & 0.981 & 0.959 \\ 
 \makecell[l]{CI Coverage  (BC-HC2, Satterthwaite)}  & & &  & 0.983 & 0.992 \\
&&&&&\\
 \textbf{DGP3.3, N =24}& & & & & \\
  Bias & 0.000 & 0.000 & -0.001 & 0.000 & 0.000 \\ 
  SD & 0.433 & 0.109 & 0.164 & 0.168 & 0.210 \\ 
  MSE & 0.433 & 0.109 & 0.164 & 0.168 & 0.210 \\ 
 \makecell[l]{CI Coverage  (HC2, Student-t)}& 0.941 & 0.942 & 0.857 & 0.817 & 0.750 \\
  \makecell[l]{CI Coverage (HC2, Satterthwaite)}  &0.948 &0.954 & 0.936  & 0.841 & 0.859 \\
  \makecell[l]{CI Coverage  (BC-HC2, Student-t)} &  & &  & 0.870 & 0.874 \\
 \makecell[l]{CI Coverage  (BC-HC2, Satterthwaite)} &  &&  & 0.896 & 0.939 \\
   \hline  
\end{tabular}
 \begin{tablenotes}
      \small
\item Note: CI Coverage measures the empirical coverage rates of nominal 95 percent confidence intervals. The unit of CI coverage is $\times$100 percentage points. CI Coverage (HC2, Student-t) is calculated using the original OLS residuals. CI Coverage (BC-HC2, Student-t) and CI Coverage (BC-HC2, Satterthwaite) are calculated using the recomputed OLS residuals. Both Student-t and Satterthwaite adjustments are calculated using the R-package \texttt{clubSandwich}\cite{clubSandwich}.
\end{tablenotes}
\end{threeparttable}
\end{table}
\begin{table}[H]
\centering
\begin{threeparttable}
\caption{Confidence Interval Statistics for Scheme 3 using HC2}
\begin{tabular}{lcccccc}
\toprule 
& \multicolumn{5}{c}{\textbf{ATE Estimators}} \\ 
\cmidrule(l){2-7} 
  & \multicolumn{3}{c}{HC2} & \multicolumn{3}{c}{BC-HC2} \\ 
  &  Z & Student-t & Satterthwaite & Z & Student-t & Satterthwaite  \\ 
  \hline
  \hline
&&&&&\\
 \textbf{DGP3.1, N =24}& & & & & \\
\textit{Coverage, Percentage}& & & & & \\
Non-interacted (Debiased) & 95.000 & 95.900 & 96.500 & 95.200 & 96.000 & 96.600 \\ 
  Interacted (Debiased) & 69.100 & 71.200 & 81.000 & 85.700 & 87.600 & 93.600 \\ 
&&&&&&\\
\textit{CI Width, Average}& & & & & \\
Non-interacted (Debiased) & 2.130 & 2.248 & 2.385 & 2.181 & 2.302 & 2.445 \\ 
  Interacted (Debiased) & 0.793 & 0.837 & 1.095 & 1.502 & 1.585 & 2.389 \\ 
 &&&&&&\\
\textit{CI Width, Median}& & & & & \\
Non-interacted (Debiased) &  2.160  & 2.280 &         2.406  & 2.205 &   2.327 &         2.449\\
  Interacted (Debiased) & 0.778  & 0.821    &      1.006 &  0.978  &  1.033 &         1.225\\
&&&&&\\
 \textbf{DGP3.2, N =24}& & & & & \\
 \textit{Coverage, Percentage}& & & & & \\
Non-interacted (Debiased) & 97.900 & 98.300 & 98.500 & 97.600 & 98.100 & 98.300 \\ 
  Interacted (Debiased) & 78.400 & 80.400 & 90.600 & 94.900 & 95.900 & 99.200 \\
&&&&&&\\
\textit{CI Width, Average}& & & & & \\
Non-interacted (Debiased) & 2.110 & 2.227 & 2.362 & 2.125 & 2.242 & 2.381 \\ 
  Interacted (Debiased) & 0.487 & 0.514 & 0.672 & 0.852 & 0.899 & 1.338 \\
&&&&&&\\
\textit{CI Width, Median}& & & & & \\
Non-interacted (Debiased) &  2.129  & 2.247 &        2.368  & 2.150   & 2.269       &  2.389\\
  Interacted (Debiased) &   0.478 &  0.505 &        0.613 & 0.575  & 0.607  &        0.721\\
&&&&&&\\
 \textbf{DGP3.3, N =24}& & & & & \\
 \textit{Coverage, Percentage}& & & & & \\
Non-interacted (Debiased)& 79.200 & 81.700 & 84.100 & 84.400 & 87.000 & 89.600 \\ 
  Interacted (Debiased)& 72.600 & 75.000 & 85.900 & 85.400 & 87.400 & 93.900 \\
&&&&&&\\
\textit{CI Width, Average}& & & & & \\
Non-interacted (Debiased) & 0.422 & 0.446 & 0.473 & 0.457 & 0.483 & 0.513 \\ 
  Interacted (Debiased) & 0.487 & 0.514 & 0.672 & 0.816 & 0.862 & 1.277 \\ 
  &&&&&&\\
\textit{CI Width, Median}& & & & & \\
Non-interacted (Debiased) &  0.417 &   0.440    &     0.465  & 0.435 &   0.459     &    0.483\\
  Interacted (Debiased) &  0.478 &  0.505    &     0.613 &  0.554 &  0.584   &      0.695\\
&&&&&&\\
   \hline
\end{tabular}
 \begin{tablenotes}
      \small
\item Note: CI Coverage measures the empirical coverage rates of nominal 95 percent confidence intervals.  Both Student-t and Satterthwaite adjustments are calculated using the R-package \texttt{clubSandwich}\cite{clubSandwich}. The number of simulation is ${24 \choose 8}$.
\end{tablenotes}
\end{threeparttable}

\end{table}
\begin{table}[H]
\centering
\begin{threeparttable}
\caption{Confidence Interval Statistics for Scheme 3 using HC3}
\begin{tabular}{lcccccc}
\toprule 
& \multicolumn{5}{c}{\textbf{ATE Estimators}} \\ 
\cmidrule(l){2-7} 
  & \multicolumn{3}{c}{HC3} & \multicolumn{3}{c}{BC-HC3} \\ 
  &  Z & Student-t & Satterthwaite & Z & Student-t & Satterthwaite  \\ 
  \hline
  \hline
&&&&&\\
 \textbf{DGP3.1, N =24}& & & & & \\
\textit{Coverage, Percentage}& & & & & \\
Non-interacted (Debiased)& 96.700 & 97.200 & 97.500 & 96.800 & 97.300 & 97.600 \\ 
  Interacted (Debiased)& 87.900 & 89.000 & 95.100 & 95.800 & 96.600 & 99.100 \\ 
 &&&&&\\
\textit{CI Width, Average}& & & & & \\
Non-interacted (Debiased) & 2.514 & 2.653 & 2.842 & 2.557 & 2.699 & 2.895 \\ 
  Interacted (Debiased)& 1.793 & 1.893 & 4.927 & 4.118 & 4.347 & 14.929 \\
&&&&&\\
 \textit{CI Width, Median}& & & & & \\
Non-interacted (Debiased)& 2.512 & 2.651 & 2.816 & 2.551 & 2.692 & 2.854 \\ 
  Interacted (Debiased) & 1.385 & 1.462 & 2.136 & 1.606 & 1.695 & 2.568 \\
&&&&&\\
 \textbf{DGP3.2, N =24}& & & & & \\ 
\textit{Coverage, Percentage}& & & & & \\
Non-interacted (Debiased)& 98.500 & 98.700 & 98.900 & 98.300 & 98.500 & 98.700 \\
  Interacted (Debiased) & 93.000 & 93.800 & 98.600 & 99.000 & 99.300 & 100.000 \\
&&&&&\\
\textit{CI Width, Average}& & & & & \\
Non-interacted (Debiased) & 2.490 & 2.628 & 2.815 & 2.491 & 2.629 & 2.818 \\ 
  Interacted (Debiased) & 0.989 & 1.044 & 2.624 & 2.280 & 2.406 & 8.289 \\ 
 &&&&&\\
 \textit{CI Width, Median}& & & & & \\
Non-interacted (Debiased) & 2.487 & 2.625 & 2.782 & 2.498 & 2.636 & 2.793 \\ 
  Interacted (Debiased)& 0.787 & 0.831 & 1.219 & 0.906 & 0.956 & 1.447 \\ 
&&&&&\\
 \textbf{DGP3.3, N =24}& & & & & \\ 
\textit{Coverage, Percentage}& & & & & \\
Non-interacted (Debiased) & 85.000 & 87.000 & 89.300 & 90.100 & 92.100 & 94.200 \\
  Interacted (Debiased) & 90.000 & 91.300 & 97.200 & 95.600 & 96.500 & 99.100 \\
&&&&&\\
\textit{CI Width, Average}& & & & & \\
Non-interacted (Debiased)& 0.489 & 0.516 & 0.552 & 0.526 & 0.555 & 0.596 \\ 
  Interacted (Debiased) & 0.989 & 1.044 & 2.624 & 2.133 & 2.251 & 7.610 \\
&&&&&\\
 \textit{CI Width, Median}& & & & & \\
Non-interacted (Debiased)& 0.479 & 0.506 & 0.538 & 0.499 & 0.526 & 0.558 \\ 
  Interacted (Debiased) & 0.787 & 0.831 & 1.219 & 0.883 & 0.932 & 1.417 \\ 
   \hline
\end{tabular}
 \begin{tablenotes}
      \small
\item Note: CI Coverage measures the empirical coverage rates of nominal 95 percent confidence intervals. HC3 refers to the approximate jackknife robust standard error estimator\cite{mackinnon1985some}. Both Student-t and Satterthwaite adjustments are calculated using the R-package \texttt{clubSandwich}\cite{clubSandwich}. The number of simulation is ${24 \choose 8}$.
\end{tablenotes}
\end{threeparttable}
\end{table}
\subsection{Scheme 4}
\begin{table}[H]
\centering
\caption{Simulation Results for Scheme 4}    
\begin{threeparttable}
\begin{tabular}{lcccccc}
\toprule 
& \multicolumn{5}{c}{\textbf{ATE Estimators}} \\ 
\cmidrule(l){2-6} 
 & Unadjusted & \multicolumn{2}{c}{OLS} & \multicolumn{2}{c}{Debiased} \\ 
  & & Non-Int. & Interacted & Non-Int. & Interacted \\ 
  \hline
  \hline
&&&&&\\
 \textbf{DGP4.1, N =24}& & & & & \\
\textit{Coverage, Percentage}& & & & & \\
Bias & 0.000 & -0.060 & -0.042 & 0.000 & 0.000 \\
  SD & 0.577 & 0.438 & 0.692 & 0.524 & 0.570 \\ 
  MSE & 0.577 & 0.442 & 0.693 & 0.524 & 0.570 \\ 
 \makecell[l]{CI Coverage  (HC2, Student-t)} & 0.862 & 0.849 & 0.801 & 0.832 & 0.803 \\ 
 \makecell[l]{CI Coverage  (HC2, Satterthwaite)} & 0.872 & 0.858  & 0.856 & 0.843 & 0.860 \\ 
  \makecell[l]{CI Coverage  (BC-HC2, Student-t)} &   && & 0.836 & 0.841 \\ 
\makecell[l]{CI Coverage  (BC-HC2, Satterthwaite)} &  &  &  & 0.847 & 0.890 \\ 
&&&&&\\
 \textbf{DGP4.2, N =24}& & & & & \\
 \textit{Coverage, Percentage}& & & & & \\
 Bias & -0.000 & 0.061 & 0.028 & 0.000 & 0.000 \\ 
  SD & 0.144 & 0.325 & 0.317 & 0.303 & 0.256 \\ 
  MSE & 0.144 & 0.331 & 0.318 & 0.303 & 0.256 \\ 
 \makecell[l]{CI Coverage  (HC2, Student-t)} & 1.000 & 0.933 & 0.907 & 0.954 & 0.930 \\ 
 \makecell[l]{CI Coverage  (HC2, Satterthwaite)} & 1.000 & 0.940  &  0.974 & 0.960 & 0.991 \\
  \makecell[l]{CI Coverage  (BC-HC2, Student-t)} &  &  &  & 0.954 & 0.952 \\ 
\makecell[l]{CI Coverage  (BC-HC2, Satterthwaite)} &  &  &  & 0.961 & 0.994 \\
&&&&&&\\
 \textbf{DGP4.3, N =24}& & & & & \\
 \textit{Coverage, Percentage}& & & & & \\
  Bias & -0.000 & 0.001 & -0.014 & 0.000 & 0.000 \\ 
  SD & 0.433 & 0.272 & 0.405 & 0.329 & 0.355 \\
  MSE & 0.433 & 0.272 & 0.405 & 0.329 & 0.355 \\ 
 \makecell[l]{CI Coverage  (HC2, Student-t)} & 0.952 & 0.948 & 0.828 & 0.895 & 0.851 \\ 
\makecell[l]{CI Coverage  (HC2, Satterthwaite)} & 0.961 & 0.960  & 0.921  & 0.916 & 0.926 \\
  \makecell[l]{CI Coverage  (BC-HC2, Student-t)}&  &  & & 0.907 & 0.868 \\ 
\makecell[l]{CI Coverage  (BC-HC2, Satterthwaite)} &  &  &  & 0.927 & 0.929 \\ 
&&&&&&\\
   \hline
\end{tabular}
 \begin{tablenotes}
     \small
\item Note: CI Coverage measures the empirical coverage rates of nominal 95 percent confidence intervals. The unit of CI coverage is $\times$100 percentage points. CI Coverage (HC2, Student-t) is calculated using the original OLS residuals. CI Coverage (BC-HC2, Student-t) and CI Coverage (BC-HC2, Satterthwaite) are calculated using the recomputed OLS residuals. Both Student-t and Satterthwaite adjustments are calculated using the R-package \texttt{clubSandwich}\cite{clubSandwich}.
\end{tablenotes}
\end{threeparttable}
\end{table}
\begin{table}[H]
\centering
\begin{threeparttable}
\caption{Confidence Interval Statistics for Scheme 4 using HC2}
\begin{tabular}{lcccccc}
\toprule 
& \multicolumn{5}{c}{\textbf{ATE Estimators}} \\ 
\cmidrule(l){2-7} 
  & \multicolumn{3}{c}{HC2} & \multicolumn{3}{c}{BC-HC2} \\ 
  &  Z & Student-t & Satterthwaite & Z & Student-t & Satterthwaite  \\ 
  \hline
  \hline
&&&&&\\
 \textbf{DGP4.1, N =24}& & & & & \\ 
\textit{Coverage, Percentage}& & & & & \\
Non-interacted (Debiased) & 82.100 & 83.200 & 84.300 & 82.400 & 83.600 & 84.700 \\ 
  Interacted (Debiased)& 79.000 & 80.300 & 86.000 & 82.800 & 84.100 & 89.000 \\
&&&&&\\
\textit{CI Width, Average}& & & & & \\
Non-interacted (Debiased)& 2.101 & 2.217 & 2.355 & 2.135 & 2.254 & 2.395 \\ 
  Interacted (Debiased) & 1.898 & 2.004 & 2.625 & 2.412 & 2.546 & 3.616 \\ 
 &&&&&\\
 \textit{CI Width, Median}& & & & & \\
Non-interacted (Debiased) & 2.166 & 2.286 & 2.419 & 2.196 & 2.318 & 2.449 \\ 
  Interacted (Debiased) & 1.822 & 1.923 & 2.293 & 1.952 & 2.060 & 2.458 \\ 
&&&&&\\
 \textbf{DGP4.2, N =24}& & & & & \\
\textit{Coverage, Percentage}& & & & & \\
Non-interacted (Debiased) & 94.600 & 95.400 & 96.000 & 94.600 & 95.400 & 96.100 \\ 
  Interacted (Debiased) & 92.000 & 93.000 & 99.100 & 94.200 & 95.200 & 99.400 \\ 
 &&&&&\\
\textit{CI Width, Average}& & & & & \\
Non-interacted (Debiased) & 1.983 & 2.093 & 2.220 & 1.990 & 2.101 & 2.229 \\ 
  Interacted (Debiased)& 1.194 & 1.260 & 1.655 & 1.424 & 1.503 & 2.103 \\
&&&&&\\
 \textit{CI Width, Median}& & & & & \\
Non-interacted (Debiased)& 2.041 & 2.154 & 2.273 & 2.051 & 2.164 & 2.283 \\ 
  Interacted (Debiased) & 1.155 & 1.219 & 1.460 & 1.193 & 1.260 & 1.497 \\ 
&&&&&\\
 \textbf{DGP4.3, N =24}& & & & & \\ 
\textit{Coverage, Percentage}& & & & & \\
Non-interacted (Debiased)& 87.500 & 89.500 & 91.600 & 88.700 & 90.700 & 92.700 \\ 
  Interacted (Debiased) & 83.100 & 85.100 & 92.600 & 85.000 & 86.800 & 92.900 \\
&&&&&\\
\textit{CI Width, Average}& & & & & \\
Non-interacted (Debiased)& 1.051 & 1.109 & 1.176 & 1.068 & 1.127 & 1.197 \\ 
  Interacted (Debiased)s & 1.194 & 1.260 & 1.655 & 1.443 & 1.523 & 2.132 \\
&&&&&\\
 \textit{CI Width, Median}& & & & & \\
Non-interacted (Debiased)& 1.028 & 1.085 & 1.148 & 1.050 & 1.108 & 1.170 \\ 
  Interacted (Debiased) & 1.155 & 1.219 & 1.460 & 1.204 & 1.271 & 1.505 \\ 
   \hline
\end{tabular}
 \begin{tablenotes}
      \small
\item Note: CI Coverage measures the empirical coverage rates of nominal 95 percent confidence intervals. Both Student-t and Satterthwaite adjustments are calculated using the R-package \texttt{clubSandwich}\cite{clubSandwich}.
\end{tablenotes}
\end{threeparttable}
\end{table}
\begin{table}[H]
\centering
\begin{threeparttable}
\caption{Confidence Interval Statistics for Scheme 4 using HC3}
\begin{tabular}{lcccccc}
\toprule 
& \multicolumn{5}{c}{\textbf{ATE Estimators}} \\ 
\cmidrule(l){2-7} 
  & \multicolumn{3}{c}{HC3} & \multicolumn{3}{c}{BC-HC3} \\ 
  &  Z & Student-t & Satterthwaite & Z & Student-t & Satterthwaite  \\ 
  \hline
  \hline
&&&&&\\
 \textbf{DGP4.1, N =24}& & & & & \\ 
\textit{Coverage, Percentage}& & & & & \\
Non-interacted (Debiased) & 85.400 & 86.500 & 87.600 & 85.600 & 86.700 & 87.900 \\ 
  Interacted (Debiased) & 89.800 & 90.500 & 95.100 & 92.200 & 92.800 & 96.600 \\ 
&&&&&\\
\textit{CI Width, Average}& & & & & \\
Non-interacted (Debiased) & 2.533 & 2.673 & 2.866 & 2.565 & 2.707 & 2.905 \\ 
  Interacted (Debiased)& 4.225 & 4.459 & 11.551 & 5.853 & 6.178 & 18.780 \\
 &&&&&\\
 \textit{CI Width, Median}& & & & & \\
Non-interacted (Debiased) & 2.561 & 2.703 & 2.880 & 2.595 & 2.739 & 2.913 \\ 
  Interacted (Debiased) & 3.143 & 3.317 & 4.784 & 3.276 & 3.458 & 5.019 \\ 
&&&&&\\
 \textbf{DGP4.2, N =24}& & & & & \\
\textit{Coverage, Percentage}& & & & & \\
Non-interacted (Debiased) & 96.300 & 97.000 & 97.600 & 96.500 & 97.100 & 97.800 \\ 
  Interacted (Debiased) & 96.900 & 97.400 & 100.000 & 98.400 & 98.800 & 100.000 \\  
 &&&&&\\
\textit{CI Width, Average}& & & & & \\
Non-interacted (Debiased) & 2.382 & 2.514 & 2.693 & 2.388 & 2.521 & 2.701 \\ 
  Interacted (Debiased) & 2.363 & 2.494 & 6.239 & 3.118 & 3.291 & 9.625 \\
&&&&&\\
 \textit{CI Width, Median}& & & & & \\
Non-interacted (Debiased)& 2.425 & 2.559 & 2.715 & 2.432 & 2.567 & 2.725 \\ 
  Interacted (Debiased)& 1.831 & 1.932 & 2.824 & 1.872 & 1.976 & 2.895 \\ 
&&&&&\\
 \textbf{DGP4.3, N =24}& & & & & \\ 
\textit{Coverage, Percentage}& & & & & \\
Non-interacted (Debiased) & 92.100 & 93.600 & 95.400 & 93.100 & 94.500 & 96.100 \\ 
  Interacted (Debiased) & 93.500 & 94.500 & 98.300 & 94.200 & 95.100 & 98.500 \\
&&&&&\\
\textit{CI Width, Average}& & & & & \\
Non-interacted (Debiased)& 1.213 & 1.281 & 1.370 & 1.228 & 1.296 & 1.388 \\ 
  Interacted (Debiased)& 2.363 & 2.494 & 6.239 & 3.258 & 3.438 & 10.277 \\
&&&&&\\
 \textit{CI Width, Median}& & & & & \\
Non-interacted (Debiased) & 1.174 & 1.239 & 1.320 & 1.198 & 1.264 & 1.343 \\ 
  Interacted (Debiased)& 1.831 & 1.932 & 2.824 & 1.891 & 1.996 & 2.930 \\ 
   \hline
\end{tabular}
 \begin{tablenotes}
      \small
\item Note: CI Coverage measures the empirical coverage rates of nominal 95 percent confidence intervals. HC3 refers to the approximate jackknife robust standard error estimator\cite{mackinnon1985some}. Both Student-t and Satterthwaite adjustments are calculated using the R-package \texttt{clubSandwich}\cite{clubSandwich}.
\end{tablenotes}
\end{threeparttable}
\end{table}
\bibliographystyle{alpha} 
\bibliography{ref}

\end{document}